\documentclass[journal,twocolumn,10pt,twoside]{IEEEtran}
\DeclareMathAlphabet\mathbfcal{OMS}{cmsy}{b}{n}
\usepackage{cite,amsmath,amssymb,mathrsfs,epsf}
\usepackage{graphicx}
\usepackage[usenames]{color}
\usepackage{multirow}
\usepackage{tabularx}
\usepackage{stfloats}
\usepackage{verbatim}
\usepackage{epsfig}
\usepackage{amsfonts}
\usepackage[normalem]{ulem}
\usepackage{amsthm}
\usepackage[table]{xcolor}
\usepackage{float}
\usepackage{paralist}
\usepackage{hyperref}
\usepackage{algorithm,algorithmic}
\usepackage{color}
\usepackage{titlesec}
\usepackage{mathtools}
\usepackage[skip=0.5\baselineskip, font=small]{caption}
\usepackage[outdir=./]{epstopdf}
\IEEEoverridecommandlockouts

\normalsize

\allowdisplaybreaks

\usepackage[caption=false,font=footnotesize,subrefformat=parens,labelformat=parens]{subfig}
\makeatletter
\def\@IEEEsectpunct{:\ \,}
\def\paragraph{\@startsection{paragraph}{4}{\z@}{1.5ex plus 1.5ex minus 0.5ex}%
{0ex}{\normalfont\normalsize\itshape\bfseries}}
\makeatother

\DeclareMathOperator{\polylog}{polylog}

\begin{document}
%
\title{Coded Caching in Networks with Heterogeneous User Activity}
%
%
%
\author{\IEEEauthorblockN{Adeel~Malik, Berksan~Serbetci, Petros~Elia}
\thanks{The authors are with the Communication Systems Department at EURECOM, Sophia Antipolis, 06410, France (email: malik@eurecom.fr, serbetci@eurecom.fr, elia@eurecom.fr). The work is supported by the European Research Council under the EU Horizon 2020 research and innovation program / ERC grant agreement no. 725929 (project DUALITY).}}


\newtheorem{axiom}{Axiom}
\newtheorem{lemma}{Lemma}
\newtheorem{corollary}{Corollary}
\newtheorem{theorem}{Theorem}
\newtheorem{proposition}{Proposition}
\newtheorem{observation}{Observation}
\newtheorem{definition}{Definition}
\newtheorem{remark}{Remark}
\newtheorem{example}{Example}
\newtheoremstyle{case}{}{}{}{}{}{:}{ }{}
\theoremstyle{case}
\newtheorem{case}{Case}

\newtheorem{problem}{Problem}[section]

\maketitle
\begin{abstract}
This work elevates coded caching networks from their purely information-theoretic framework to a stochastic setting, by exploring the effect of random user activity and by exploiting correlations in the activity patterns of different users. In particular, the work studies the $K$-user cache-aided broadcast channel with a limited number of cache states, and explores the effect of cache state association strategies in the presence of arbitrary user activity levels; a combination that strikes at the very core of the coded caching problem and its crippling subpacketization bottleneck. We first present a statistical analysis of the average worst-case delay performance of such subpacketization-constrained (state-constrained) coded caching networks, and provide computationally efficient performance bounds as well as scaling laws for any arbitrary probability distribution of the user-activity levels. The achieved performance is a result of a novel user-to-cache state association algorithm that leverages the knowledge of probabilistic user-activity levels. 

We then follow a data-driven approach that exploits the prior history on user-activity levels and correlations, in order to predict interference patterns, and thus better design the caching algorithm. This optimized strategy is based on the principle that users that overlap more, interfere more, and thus have higher priority to secure complementary cache states. This strategy is proven here to be within a small constant factor from the optimal. Finally, the above analysis is validated numerically using synthetic data following the Pareto principle. To the best of our understanding, this is the first work that seeks to exploit user-activity levels and correlations, in order to map future interference and design optimized coded caching algorithms that better handle this interference.
\end{abstract}


\IEEEpeerreviewmaketitle

%
\IEEEpeerreviewmaketitle

 \setlength{\abovedisplayskip}{3pt}   
 \setlength{\belowdisplayskip}{3pt}    
 \setlength{\skip\footins}{0.5cm}

\section{Introduction}
\IEEEPARstart{T}{he} volume of mobile data traffic is rapidly growing, and soon existing networks will not have enough bandwidth resources to support this dramatically increasing demand~\cite{cisco2020}. In this context, caching offers a promising means of increasing efficiency by proactively storing part of the data at the network edge~\cite{shanmugam_it13}, including at wireless communication stations as well as on end-user devices~\cite{golrezaei_wc14,ji_jsac16}.

While generally caching is based on the idea that storing data can allow a receiving node to have easy access to \emph{its own} desired file, recent work has shown the powerful effects of exploiting the existence of the aforementioned desired file \emph{at the caches of other receiving users}. In interference-limited scenarios --- such as in downlink settings exemplified by the broadcast channel where each user has access to their own cache and requires their own distinct file --- the findings in~\cite{man_it14} suggest that a proper use of caching can allow for single multicast transmissions to simultaneously serve many users each having their own distinct demands. This breakthrough in the way caching is perceived, is based on the ideas of index coding which tells us that when the stored content in one user's cache overlaps with other users' requests, one can design multicast transmissions (in the form of XORs or other linear combinations of desired data), that allow for rapid delivery of any possible set of demands. Index coding --- which is generally a computationally hard problem~\cite{langberg2011hardness} ---  has received significant attention in a literature that has explored its performance limits~\cite{bar2011index, maleki2014index}, as well as its strong connections to the network coding problem~\cite{fragouli2007network, effros2015equivalence}. One main difference between index coding and network coding is that index coding specializes on cache-related cases in the sense that it considers receivers that benefit from side-information, which in our case can be found, for example, in the caches.

Motivated by index coding and its ability to exploit receivers' side information to create coded multicasting opportunities for users requesting different files, the seminal work in~\cite{man_it14} has introduced the concept of \emph{coded caching}. This work revealed that --- under some theoretical assumptions, and in the presence of a deterministic information-theoretic broadcast framework ---  the use of caching at the receivers can allow the simultaneous delivery of an unlimited number of user-requests, with a limited delay. This astounding conclusion was achieved by carefully designing a combinatorial clique-based cache placement algorithm, and a synergistic delivery scheme that enables transmitting independent content to multiple users at a time. In essence, coded caching associates each receiving user to its own cache state in a manner that allows for the custom design of a long sequence of high-capacity index coding problems that are served one after the other. As we will see soon though, this requirement that each user has their own cache state, is an assumption that --- in essence --- cannot hold. 

To see this, let us quickly recall that in its original setting, coded caching considers a unit-capacity single-stream broadcast channel (BC), where a transmitting base station (BS) has access to a library (catalog) of $N$ unit-sized files, and serves $K$ receiving users each equipped with a cache of size equal to the size of $M$ files, or equivalently equal to a fraction $\gamma = \frac{M}{N}$ of the library. In this context, the work in~\cite{man_it14} provides a novel placement and delivery scheme that can serve any set of $K$ simultaneous requests with a worst-case delivery time of $T = \frac{K(1-\gamma)}{1+K\gamma} \approx \frac{1-\gamma}{\gamma}$. This ability to serve a theoretically ever-increasing number of users with a bounded delay, is a direct result of exploiting the cache-enabled multicasting opportunities that allow for delivery to $K\gamma + 1$ users at a time.

As suggested above though, coded caching has a serious Achilles' heel. In particular, for the above performance to be guaranteed, coded caching requires that each user be allocated their own specifically-designed cache state (cache content), which --- without delving into the esoteric details of coded caching --- effectively requires the partitioning of each library-file into $K \choose K\gamma$ subpackets. This number scales exponentially in $K$, and thus requires files to be of truly astronomical sizes. Thus given any reasonable constraint on the file sizes, the number of cache states is effectively forced to be reduced, and the aforementioned coding gains are indeed diminished to gains that are considerably less than $K\gamma+1$. What this file-size constraint (also known as the subpacketization bottleneck) effectively forces is the reduction of the number of cache states\footnote{This simply means that even though there are $K$ different users, each with their own physical cache, in essence, there can only exist $\Lambda$ distinct caches, that must be shared among the users. This effectively means that groups of users are forced to have identical, rather than complementary, cache contents.} to some $\Lambda\ll K$, which --- under the basic principles of the clique-based cache-placement in~\cite{man_it14} --- allows for a smaller subpacketization level $\Lambda \choose \Lambda\gamma $ $\ll$ $ K \choose K\gamma$ at the expense though of a much reduced coding gain $\Lambda\gamma+1 \ll K\gamma+1 $ and a much larger delay $T = \frac{K(1-\gamma)}{1+\Lambda\gamma}$ which is now unbounded. For more details on this, the reader can refer to~\cite{shanmugam_it16,lampiris_jsac18}.

\subsection{The connection between coded caching, complementary cache states, user-activity levels and user-activity correlations}
The performance of coded caching in the presence of an inevitably reduced number of cache states, has been explored in various works that include the work in~\cite{jin_toc_19} which introduced a new scheme for this setting, and the work in~\cite{parrinello_it20} which established the fundamental limits of the state-limited coded caching setting, by deriving the exact optimal worst-case delivery time as a function of the user-to-cache state association profile that represents the number of users served by each cache.

As we witness in the above works, in order to maintain the ability to jointly exploit multicasting opportunities, users must be associated to complementary cache states that are carefully designed and which cannot be identical. The above findings reveal that a basic problem with the state-limited scenario (where $\Lambda \ll K$) in coded caching is simply the fact that if two or more users are forced to share the same cache state (i.e., the same content in their caches), then these users generally do not have the ability to jointly receive a multicasting message that can be useful to all.  Such state-limited scenario results in the aforementioned large deterioration in performance, irrespective of the user-to-cache association policy. What we additionally learn from the work in~\cite{malik_tcom21} is that if the users are assigned states at random, then this randomness imposes an additional \emph{unbounded} performance deterioration that is a result of `unfortunate' associations where too many users share the same cache state. That is why the task of user-to-cache state association is important. 

At the same time though, coded caching experiences a certain synchronization aspect, which is a direct outcome of the fact that users are expected to be partially asynchronous in their timing of requesting files. Hence, the notion of time is of essence. This asynchronicity has a negative aspect, but also a positive one; both of which we explore here. On the one hand, having only a fraction of the users appear simultaneously, implies a smaller number of users that can simultaneously participate in coded caching and thus implies potentially fewer multicasting opportunities and thus a smaller coding gain. On the other hand, such asynchronicity implies less instantaneous interference. 
This is where user activity levels come into the picture, and this is where user activity correlations can be exploited. In essence --- as it will become clearer later on --- any users that are correlated in terms of their activity in time, should be associated to different cache states, as this is essential in using caches for handling their mutual interference. On the other hand, knowing that some users rarely request data at the same time, allows us to give them the same cache state resource. In essence, users that overlap more, interfere more, and thus have higher priority to secure complementary cache states. This optimization effort is particularly important because, as we recall, these resources are indeed scarce. By exploring user activities and learning from their history, we are able to predict interference patterns, and then we are able to assign cache states accordingly. This is, to the best of our understanding, the first work that seeks to exploit user-activity levels and correlations, in order to map future interference and provide optimized caching algorithms that better handle this interference.

\subsection{Network setting}
We consider a cache-aided wireless network, which consists of a base station and $K$ cache-enabled receiving users. The base station (BS) has access to a library of $N$ equisized files $\mathcal{F}=\left[F_1, F_2, \dots, F_N \right]$ and delivers content via a broadcast link to $K$ receiving users. Each user $k \in \left[1, 2, \dots, K \right]$ is equipped with a cache of normalized storage capacity of $\gamma \triangleq \frac{M}{N} \in [0,1]$, and requests a file from the content library with probability $p_k$. We use $\textbf{p}=  \left[p_1, p_2, \dots,p_K\right]$ to denote the users \emph{activity level vector}. At any instance, if a user $k$ is requesting a file, then we say that the user $k\in [K]$ is an \emph{active user}. Naturally $K_{\mathbf{p}}= \sum_{k=1}^K p_k$ is the expected number of active users. Figure~\ref{fig:SM} depicts an instance of our cache-aided wireless network.
 \begin{figure}[t]
\centering
\includegraphics[width=.7\linewidth]{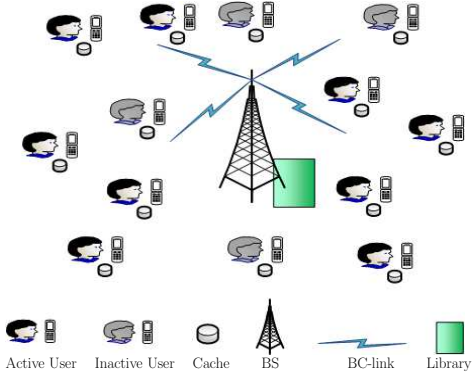}
\caption{An instance of a cache-aided wireless network.}\label{fig:SM}
\end{figure}


The communication process consists of two phases; the \emph{placement phase} and the \emph{delivery phase}. During the placement phase, each user's cache is filled with the content from the library, and this phase is oblivious to the upcoming number of users in the delivery phase, as well as is oblivious to the upcoming file demands. The delivery phase begins with \emph{the active users} simultaneously requesting one file each, and continues with the BS delivering this content to the users. This phase is naturally aware of the demands of the active users, as well as is aware of the content cached at each user.


\noindent\textbf{\emph{Placement phase}:} 
We consider the subpacketization-constrained uncoded cache placement scheme based on~\cite{man_it14}. Let $B_{max}$ denote the maximum allowable subpacketization of a file, which defines the maximum number of cache states as follows
\begin{align*}
  \Lambda= \arg \underset{k\leq K}{ \max} \left\{ {k \choose k\gamma} \leq B_{max}\right\}.  
\end{align*}
Each file $F_i \in \mathcal{F}$ is partitioned into ${\Lambda \choose t}$ distinct equisized subpackets, where $t \triangleq \Lambda\gamma$ for some $t \in [1, \dots, \Lambda]$. Then we index each subpacket of a file by a distinct subset $\tau \subseteq [1, \dots, \Lambda]$ of size $t$. The set of indexed subpackets corresponding to file $F_i \in \mathcal{F}$ is given by $\left\{F_{i,\tau}: \tau \subseteq [1, \dots, \Lambda], \left|\tau\right|=t\right\}$.  The content corresponding to each cache state $ \lambda \in [1, \dots, \Lambda]$ is then given by
\begin{align*}
C_{\lambda} =   \left\{F_{i,\tau}: i\in[1, \dots, N], \lambda \in \tau, \tau \subseteq [1, \dots, \Lambda], \left|\tau\right|=t\right\},
\end{align*}
where each cache state consists of $\left|C_{\lambda} \right| = N{\Lambda-1 \choose t-1}$ subpackets, which abides by the cache-size constraint since $N\frac{{\Lambda-1 \choose t-1}}{{\Lambda \choose t}} = M$.


During the placement phase, each user's cache is filled with the content of one of the cache states $\lambda \in [1, \dots, \Lambda]$. The employed user-to-cache state association is defined by a matrix $\mathbf{G}=  [0,1]^{\Lambda \times K}$, of which the $(\lambda,k)$ element $g_{\lambda,k}$ takes the value 1 if user $k$ is storing the content of cache state $\lambda \in [\Lambda]$, else $g_{\lambda,k} = 0$. We denote by $\mathbf{G}_{\lambda}$ the set of users caching the content of cache state $\lambda\in [1, \dots, \Lambda]$.

\noindent\textbf{\emph{Delivery phase}:} The delivery phase commences with each active user requesting a single file from the content library. In line with the common assumptions in coded caching, we assume that requests are generated simultaneously by active users, and that each active user requests a different file.  During this phase, the BS is aware of the user-to-cache state association matrix $\mathbf{G}$. Once the BS receives the users' requests, it commences delivery of the coded subpackets over a unit-capacity\footnote{Here the capacity is measured in units of file.} error-free broadcast link. Here, together with the aforementioned optimal placement, we also consider the optimal\footnote{Optimality here refers to the performance of the scheme over the traditional (deterministic) coded caching problem with \emph{constant} user activity.} multi-round delivery scheme of~\cite{jin_toc_19,parrinello_it20}. At any instance of the problem, the \emph{cache load vector} is denoted by $\mathbf{V}= \left[v_1,  \dots, v_{\Lambda} \right]$, where $v_{\lambda}$ represents the number of active users that are associated with cache state $\lambda \in [\Lambda]$.  Additionally, we use $\mathbf{L}= \left[l_1,  \dots, l_{\Lambda} \right] = sort(\mathbf{V})$ to be the \emph{profile vector}, which is the sorted (in descending order) version of the cache load vector $\mathbf{V}$.

\subsection{Metrics of interest}
To capture the randomness in user activity, we consider --- for any given user-to-cache state association matrix $\mathbf{G}$ --- the averaging metric
\begin{equation}
\overline{T}(\mathbf{G}) \triangleq E_{\mathbf{V}}[T(\mathbf{V})] = \sum_{\mathbf{V}}P(\mathbf{V})T(\mathbf{V}),
\end{equation}
where $P(\mathbf{V})$ is the probability of $\mathbf{V}$, and where $T(\mathbf{V})$ is the worst-case delivery time\footnote{The time scale is normalized such that a unit of time corresponds to the optimal amount of time needed to send a single file from the BS to the user, had there been no caching and no interference.} needed to complete the delivery of requested files given a certain cache load vector $\mathbf{V}$ associated to matrix $\mathbf{G}$. For any cache load vector $\mathbf{V}$ such that $sort(\mathbf{V}) = \mathbf{L}$, the information-theoretically optimal delivery time --- achieved with the multi-round delivery scheme~\cite{jin_toc_19,parrinello_it20} --- takes the form
\begin{align}\label{eq:TL}
 T(\mathbf{L})  =    \sum_{\lambda=1}^{\Lambda-t} l_{\lambda} \frac{{\Lambda-\lambda \choose t} }{{\Lambda \choose t}}.
 \end{align}
Thus, the average delay takes the form
\begin{align}\label{eq:tavg}
\overline{T}(\mathbf{G})= \sum_{\mathbf{L} \in \mathcal{L} }P(\mathbf{L})  T(\mathbf{L}) &=   \sum_{\lambda=1}^{\Lambda-t} \sum_{\mathbf{L} \in \mathcal{L} }P(\mathbf{L}) l_{\lambda}     \frac{ {\Lambda-\lambda \choose t} }{{\Lambda \choose t}}\nonumber\\  &=   \sum_{\lambda=1}^{\Lambda-t} E[ l_{\lambda} ]   \frac{ {\Lambda-\lambda \choose t} }{{\Lambda \choose t}},
 \end{align}
\noindent where $\mathcal{L}$ describes the set of all possible profile vectors $\mathbf{L}$, where $P(\mathbf{L})$ is the probability of a profile vector $\mathbf{L}$ \emph{given the user-to-cache state association} $\mathbf{G}$, and where $E[l_{\lambda} ]$ is the expected number of active users in the $\lambda$-th most loaded cache, again given $\mathbf{G}$.



Our interest is in finding the optimal user-to-cache association that minimizes the average delay. This corresponds to the following optimization problem.
\begin{problem} \label{eq:opt_P1}
\begin{align} 
\underset{\mathbf{G}} {\min} \ \ \ \overline{T}(\mathbf{G})
\end{align}
subject to
\begin{align}\label{eq:Cons}
 &\sum_{i=1}^{\Lambda}g_{i,k}=1 \ \ \forall   k\in [K].
\end{align}
 \end{problem}
\subsection{Our contribution}~\label{sec:contribution}
In this work, we analyze a state-constrained coded caching network of $K$ cache-aided users with $\Lambda$ cache states, when users have different activity levels, and the association between users and cache states is subject to an arbitrary grouping strategy $\mathbf{G}$. Our aim is to provide analytical bounds on the performance. We will do so either in a manner that is numerically tractable, or in the form of asymptotic approximations that offer direct insight. The following are our contributions, step by step.
\begin{itemize}
\item In Section~\ref{sec:results_arb}, for any arbitrary user activity level vector $\mathbf{p}$ and any arbitrary user-to-cache state association $\mathbf{G}$,
\begin{itemize}
\item We derive a bound on the average delay $\overline{T}(\mathbf{G})$. This bound can be evaluated in a computationally efficient manner.
\item We characterize the scaling laws of $\overline{T}(\mathbf{G})$ which take clear and insightful forms.
\item Based on the insights from the derived bounds, we propose a new user-to-cache association algorithm that seeks to minimize the average delay.
\end{itemize}
\item In Section~\ref{sec:results_iden}, we analyze the special case of uniform user activity statistics, and uniform user-to-cache state association $\mathbf{G}$. For this setting, we provide analytical upper and lower bounds on the performance, and show that the bounds have a bounded gap between them and thus a bounded gap to the optimal. Then, we proceed to characterize the exact scaling laws of $\overline{T}(\mathbf{G})$.
\item In Section~\ref{sec:results_data}, we extend our analysis to the data-driven setting, where --- in designing the caching policy --- we are able to learn from the past $S$ different demand vectors. Using this bounded-depth user-request history, we propose a heuristic user-to-cache state association algorithm which is simple to implement and which we prove here to be at most at a factor of $\frac{\log{S}}{\log\log{S}}$ from the optimal. This factor, as we argue later below, remains less than 3-4 for any reasonable scenario\footnote{If we consider a scenario where we assign caches to users once a day, and assuming that independent demand vectors appear once every $30$ minutes, then the number $S$ is at most $2\times 24 = 48$ which implies a gap of approximately $2.3$. If instead we assign caches once a week, $S$ becomes $7\times 2\times 24 = 336$ and the gap is approximately $2.7$. If this depth changes to a much larger $S = 12\times336$ corresponding to a history window of 4 weeks, and a demand vector --- for those same $K$ co-located users --- every 10 minutes, then the gap is bounded at 3.3.}, which is validated numerically using synthetic data following the Pareto principle.
\item In Section~\ref{sec:validation}, we perform extensive numerical evaluations that validate our analysis.
\end{itemize}

\subsection{Notations}~\label{sec:notations}
Throughout this paper, for $n$ a positive integer, we use the notation $[n]\triangleq\left[1, 2, \dots, n \right], \forall n \in \mathbb{Z}^+$. We use $\mathbf{A} / \mathbf{B}$ to denote the difference set that consists of all the elements of set $\mathbf{A}$ not in set $\mathbf{B}$. Unless otherwise stated, logarithms are assumed to have base 2. We also use the following asymptotic notation: i) $f(x)= O(g(x))$ will mean that there exist constants $a$ and $c$ such that $f(x) \leq ag(x), \forall x > c$, ii) $f(x) = o(g(x))$ will mean that  $\lim_{x \rightarrow \infty  }\frac{f(x)}{g(x)} =0 $, iii) $f(x)=\Omega(g(x))$ will be used if $g(x)= O(f(x))$, iv) $f(x) = \omega(g(x))$ will mean that  $\lim_{x \rightarrow \infty  }\frac{g(x)}{f(x)} =0 $, and finally v) $f(x)=\Theta(g(x))$ will be used if $f(x)= O(g(x))$ and $f(x)=\Omega(g(x))$.  We use the term $\polylog(x)$ to denote the class of functions $\bigcup_{k\geq 1} O((\log x)^k)$ that are polynomial in $\log x$.

\section{Main Results: Statistical Approach}~\label{sec:results}
In this section, we present our main results on the performance of a coded caching network of $K$ cache-aided users and $\Lambda$ cache states, where the user-activity levels follow an arbitrary probability distribution $\mathbf{p}$ and where the association between users and cache states is subject to an arbitrary association strategy $\mathbf{G}$.

We can see from~\eqref{eq:tavg} that for any given user-to-cache state association $\mathbf{G}$ and user-activity statistics $\mathbf{p}$, the exact evaluation of~\eqref{eq:tavg} is computationally expensive especially for large system parameters, as the creation of $\mathcal{L}$ is an integer partition problem, and the cardinality of $\mathcal{L}$ is known to be growing exponentially with system parameters $K$ and $\Lambda$~\cite{stojmenovic_ijcm98}.
Motivated by this complexity, we here proceed to provide computationally efficient bounds on the performance. After doing so, we resort to asymptotic analysis of the impact of $\mathbf{G}$ and $\mathbf{p}$ on the performance, and provide an insightful characterization of the scaling laws of this performance. Finally, based on the insights from these scaling laws, we propose a heuristic user-to-cache state association algorithm that aims to minimize the worst-case delivery time.

\subsection{Performance analysis with arbitrary activity levels}~\label{sec:results_arb}
In this subsection, we present the statistical analysis of our problem for the general setting of an arbitrary user-to-cache state association strategy $\mathbf{G}$ and an arbitrary activity level vector $\mathbf{p}$. Crucial to our analysis for this setting will be the mean $\mu_{\lambda}= \sum_{k\in G_{\lambda}} p_{k}$ and the variance $\sigma^2_{\lambda}= \sum_{k\in G_{\lambda}} p_{k} (1-p_{k})$ of the number of active users that are caching the content of cache state $\lambda \in [\Lambda]$. Now we proceed to present our first result which is the characterization of faster-to-compute analytical bounds on the performance.
\begin{theorem}\label{th:UBLB}
In a state-constrained coded caching network of $\Lambda$ cache states, $K$ cache-aided users with normalized cache capacity $\gamma$ and activity level vector $\mathbf{p}$, the average delay $\overline{T}(\mathbf{G})$ for a given user-to-cache state association strategy $\mathbf{G}$ is bounded as follows 
\begin{equation}\label{eq:UBtavg}
\overline{T}(\mathbf{G})\!\leq\!\frac{\Lambda\!-\!t}{1\!+\!t}\! \left(\!A\!-\!\sum_{x=0}^{A-1} \max\! \left(\!0,1\!-\!\Lambda\!+\!\sum_{\lambda=1}^{\Lambda} F_1(\lambda, x)\! \right)\!\right), \end{equation}
 \begin{align} \label{eq:LBtavg}
   \overline{T}(\mathbf{G}) &\geq \frac{\Lambda-t}{1+t} \frac{t}{ \Lambda-1}\left( A- \sum_{x=0}^{A-1} \frac{\sum_{\lambda=1}^{\Lambda} F_2(\lambda,x) }{\Lambda} \right)\nonumber\\ &+\frac{\Lambda-t}{1+t}\frac{K_{\mathbf{p}}}{\Lambda}  \frac{  \Lambda-t-1}{\Lambda-1},
 \end{align}
where $t=\Lambda \gamma$, $A=\max\left(\left\{|\mathbf{G}_{\lambda}|\right\}_{\lambda=1}^{\Lambda}\right)$, where $\mathbf{G}_{\lambda}$ is the set of users caching the content of cache state $\lambda$,  
\begin{align}\label{eq:lbcdf}
F_1(\lambda,x) = 
\begin{cases}\! 0  \hspace{2.7cm} \text{ \emph{if} } 0 \leq x\leq \mu_{\lambda}-1    \\
F_{bin}\left(|\mathbf{G}_{\lambda}|, \frac{\mu_{\lambda}}{|\mathbf{G}_{\lambda}|}, x\right)\hspace{0.2cm} \text{ \emph{if} } \mu_{\lambda} \leq x\leq |\mathbf{G}_{\lambda}|\\
1 \hspace{2.7cm} \text{ \emph{if} }  x> |\mathbf{G}_{\lambda}|,
\end{cases}
\end{align}
 \begin{align}\label{eq:ubcdf}
F_2(\lambda,x) = 
\begin{cases}
F_{bin}\left(|\mathbf{G}_{\lambda}|, \frac{\mu_{\lambda}}{|\mathbf{G}_{\lambda}|}, x\right) \hspace{0.2cm} \text{ \emph{if} }  0 \leq x\leq \mu_{\lambda} -1 \\
1 \hspace{2.7cm} \text{ \emph{if} }  x> \mu_{\lambda}-1,
\end{cases}
\end{align}
where $F_{bin}\left(n, q, x\right) = \sum_{i=0}^{x} {n \choose i} q^i \left(1-q\right)^{n-i}$ and where $\mu_\lambda=\sum_{k\in \mathbf{G}_{\lambda}} p_{k}$.
\end{theorem}
\begin{proof} 
The proof is deferred to Appendix~\ref{AP:UBLB}. 
\end{proof}

\begin{remark}~\label{R:1}
The bounds in Theorem~\ref{th:UBLB} can be computed in a computationally-efficient manner, as for each $\lambda \in [\Lambda]$, their evaluation only requires to compute the binomial cumulative distribution function  $F_{bin}\left(|\mathbf{G}_{\lambda}|, \frac{\mu_{\lambda}}{|\mathbf{G}_{\lambda}|}, x\right)$ for all $x\in [0, 1, 2, \cdots,|\mathbf{G}_{\lambda}|]$ of a random variable with $|\mathbf{G}_{\lambda}|$ independent trials and $\frac{\mu_{\lambda}}{|\mathbf{G}_{\lambda}|}$ success probability\footnote{In theory, $F_{bin}\left(|\mathbf{G}_{\lambda}|, \frac{\mu_{\lambda}}{|\mathbf{G}_{\lambda}|}, x\right)$ needs to be calculated for all values of $x\in\left[0,|\mathbf{G}_{\lambda}|\right]$. However, it is known that there exists a $\tilde{x} \in\left[0,|\mathbf{G}_{\lambda}|\right]$, where $F_{bin}\left(|\mathbf{G}_{\lambda}|, \frac{\mu_{\lambda}}{|\mathbf{G}_{\lambda}|}, \tilde{x}\right) \approx 1$. By De Moivre-Laplace Theorem, it is known that binomial distribution can be approximated by the normal distribution in the limit of large $|\mathbf{G}_{\lambda}|$, and the well-known 68–95–99.7 rule states that $\tilde{x} << |\mathbf{G}_{\lambda}|$. Since $F_{bin}\left(|\mathbf{G}_{\lambda}|, \frac{\mu_{\lambda}}{|\mathbf{G}_{\lambda}|}, x\right) \approx 1$ for
any $x \geq \tilde{x}$, both~\eqref{eq:lbcdf} and~\eqref{eq:ubcdf}, and consequently~\eqref{eq:UBtavg} and~\eqref{eq:LBtavg} can be quickly evaluated with high accuracy.}.
\end{remark}

Next, we proceed to our next result, which provides the asymptotic analysis of the average delay $\overline{T}(\mathbf{G})$, in the limit of large $\Lambda$ and $K$. Let us quickly recall that $\mu_{\lambda}= \sum_{k\in G_{\lambda}} p_{k}$ and $\sigma^2_{\lambda}= \sum_{k\in G_{\lambda}} p_{k} (1-p_{k})$ are respectively the mean and variance of the number of active users that are associated with cache state $\lambda$.
\begin{theorem}\label{th:lawtavg}
In a state-constrained coded caching network of $\Lambda$ cache states, $K$ cache-aided users with normalized cache capacity $\gamma$ and activity level vector $\mathbf{p}$, the average delay $\overline{T}(\mathbf{G})$ for a given association strategy $\mathbf{G}$ scales as
\begin{align}\label{eq:lawtavg_ub}
\overline{T}(\mathbf{G})  =   O\left( \left(\frac{K_{\mathbf{p}}}{\Lambda}  + \sqrt{\sum_{i=1}^{\Lambda} (\sigma_i^2+ (\mu_i-\mu)^2) }\right) \frac{\Lambda-t}{1+t}\right),
 \end{align}
 and
 \begin{align} \label{eq:lawtavg_lb}
\overline{T}(\mathbf{G}) = \Omega\left(  \frac{K_{\mathbf{p}}}{\Lambda} \frac{\Lambda-t}{1+t}\right),
 \end{align}
where $\mu= \frac{1}{\Lambda}\sum_{\lambda=1}^{\Lambda} \mu_{\lambda}= \frac{K_{\mathbf{p}}}{\Lambda}$.
\end{theorem}
\begin{proof} 
This proof is deferred to Appendix~\ref{AP:lawtavg}.
\end{proof}



Furthermore we have the following. 
\begin{corollary}~\label{cor:optG}
Any association strategy $\mathbf{G}$ that satisfies $\sqrt{\sum_{i=1}^{\Lambda} (\sigma_i^2+ (\mu_i-\mu)^2) }= O\left(\frac{K_{\mathbf{p}}}{\Lambda}\right)$ is order-optimal.
\end{corollary}
\begin{proof}
Since the lower bound in \eqref{eq:lawtavg_lb} is independent of the association strategy $\mathbf{G}$, this implies that the optimal average delay $\overline{T}^*$ (corresponding to an optimal association $\hat{\mathbf{G}}$) is lower bounded by 
\begin{align} \label{eq:lawtavgopt_lb}
\overline{T}^* = \Omega\left(\frac{K_{\mathbf{p}}(1-\gamma)}{1+t}\right).
 \end{align}
Therefore any association $\mathbf{G}$ for which the gap factor $\sqrt{\sum_{i=1}^{\Lambda} (\sigma_i^2+ (\mu_i-\mu)^2) }$ scales as $O\left(\frac{K_{\mathbf{p}}}{\Lambda}\right)$ would be order optimal as the scaling order of~\eqref{eq:lawtavg_ub} yields $O\left(  \frac{K_{\mathbf{p}}(1-\gamma)}{1+t}\right)$, thus, giving the exact scaling law of $\overline{T}(\mathbf{G}) = \Theta\left(  \frac{K_{\mathbf{p}}(1-\gamma)}{1+t}\right)$.

\end{proof}

Following the insights from Corollary \ref{cor:optG}, we now propose an algorithm that solves Problem~\ref{eq:opt_P1}.
\subsubsection{Algorithm \ref{alg:algo_p}:} The algorithm aims to heuristically minimize $\sum_{i=1}^{\Lambda} (\sigma_i^2+ (\mu_i-\mu)^2)$, and it works in $K$ iterations, where for each iteration the algorithm finds a user and cache state pair $(\hat{k},\hat{\lambda})$ in accordance to step 02 of this algorithm. Consequently user $\hat{k}$ is assigned cache state $\hat{\lambda}$.
\begin{algorithm}
  \caption{}
  \label{alg:algo_p}
  \begin{algorithmic}
 \STATE   \textbf{Input:} $\textbf{p}$, $K$, and $\Lambda$ \\
\STATE \textbf{Output:} $\mathbf{G}$ \\
\STATE \textbf{Initialization:} $\mathbf{G} \leftarrow 0$; $\mathcal{K} \leftarrow [K]$ \\
\STATE Step 01: \textbf{for} $j$ \textbf{from} 1 \textbf{to} $K$ \textbf{do} \\
\STATE Step 02: \ \ $[\hat{\lambda},  \hat{k}]\!\! \leftarrow  \arg\!\!\!\!\!\!\min\limits_{\lambda \in [\Lambda], k \in \mathcal{K}  }\sum_{i=1}^{\Lambda} (\sigma_i^2+ (\mu_i-\mu)^2)$ \\
\STATE Step 03: \ \ $g_{\hat{\lambda},\hat{k}} \leftarrow 1$ \\
\STATE Step 04: \ \ $\mathcal{K} \leftarrow \mathcal{K} \symbol{92} \hat{k}$ \\
\STATE Step 05: \textbf{end for} \\
  \end{algorithmic}
\end{algorithm}

In Section~\ref{sec:validation}, we will verify that the bounds presented in Theorem~\ref{th:UBLB} are valid for any user-to-cache state association strategy. We will also show that Algorithm~\ref{alg:algo_p} provides an efficient user-to-cache state association that yields a performance very close to the performance of optimal user-to-cache state association.


\subsection{Performance analysis with uniform activity level}~\label{sec:results_iden}
In this subsection, we analyze a special setting where users have a uniform activity level $p$, corresponding to the equiprobable case of $p_1= p_2 \dots =p_K= p$. As is common, we will also assume that $I\triangleq\frac{K}{\Lambda}$ is an integer.

\begin{lemma}\label{lem:identical_uniform}
In the presence of uniform activity level probabilities $p$, the optimal user-to-cache state association policy is the uniform one where each cache state is allocated to $K/\Lambda$ users. 
\end{lemma}
\begin{proof}
We first note that \eqref{eq:TL} implies that the average delay is minimized when $\mathbf{L}$ is uniform. For the case where we have uniform activity levels, setting $|G_{\lambda}|=I$ for all $\lambda \in [\Lambda]$, results in $P(\mathbf{L})$ being maximized for uniform $\mathbf{L}$. 
\end{proof}

We now proceed to provide computationally efficient analytical bounds on the average delay $\overline{T}(\mathbf{G})$ achieved by the uniform association policy, and subsequently to provide the exact scaling laws of this policy.
\begin{theorem}\label{th:UBLB_iden}
In a state-constrained coded caching network of $\Lambda$ cache states, $K$ cache-aided users with normalized cache capacity $\gamma$ and activity level of $p$, the average delay $\overline{T}(\mathbf{G})$ corresponding to the uniform user-to-cache state association strategy $\mathbf{G}$ is bounded by
\begin{align}\label{eq:UBtavg_iden}
\overline{T}(\mathbf{G}) &\leq    \frac{\Lambda-t}{1+t} E[l_1] \end{align}
and 
 \begin{align} \label{eq:LBtavg_iden}
 \overline{T}(\mathbf{G}) \geq \frac{\Lambda-t}{1+t}  \left( \frac{ E[ l_{1} ] t  }{ \Lambda-1} + \frac{ K p}{\Lambda} \frac{ \Lambda-t-1}{\Lambda-1} \right)
 \end{align}
where 
\begin{align}
E[l_1] &= I- \sum_{j=0}^{I-1}\left(\sum_{i=0}^{j} {I \choose i} p^i \left(1-p\right)^{I-i}\right)^{\Lambda}.
\end{align}

\end{theorem}
\begin{proof} 
The proof is deferred to Appendix~\ref{AP:UBLB_iden}. 
\end{proof}

Furthermore, the following shows that the bounds remain relatively close to the exact $\overline{T}(\mathbf{G})$. 
\begin{corollary}~\label{c:gapb}
For any fixed $\gamma\leq 1-\frac{1}{\Lambda}$, the multiplicative gap between the analytical upper bound (AUB) in~\eqref{eq:UBtavg_iden} and the analytical lower bound (ALB) in~\eqref{eq:LBtavg_iden}, is at most $\frac{\Lambda-1}{t} < 1/\gamma$. This allows us to identify the exact $\overline{T}(\mathbf{G})$ within a factor that is independent of both $\Lambda$ as well as $K$.
\end{corollary}
\begin{proof} The proof follows directly from the fact that $\frac{\Lambda-t}{1+t} \frac{E[ l_{1} ] t}{\Lambda-1} \leq \overline{T}(\mathbf{G}) \leq  \frac{\Lambda-t}{1+t} E[ l_{1} ]$.
\end{proof}

\begin{remark}
We note that the range of $\gamma\leq 1-\frac{1}{\Lambda}$ covers in essence the entire range of $\gamma$ and most certainly covers the range of pertinent $\gamma$ values.
\end{remark}

We now proceed to exploit the bounds in Theorem \ref{th:UBLB_iden}, in order to provide in a simple and insightful form, the exact scaling laws of performance. The following theorem provides the asymptotic analysis of the average delay $\overline{T}(\mathbf{G})$, in the limit of large $\Lambda$ and $K$.
\begin{theorem}\label{th:lawtavg_iden}
In a coded caching setting with $\Lambda$ cache states and $K$ cache-aided users with equal cache size $\gamma$ and activity level $p$, the average delay $\overline{T}(\mathbf{G})$ corresponding to the uniform association strategy $\mathbf{G}$ scales as
\begin{align}\label{eq:lawtavg_iden}
\overline{T}(\mathbf{G}) =
 \begin{cases} 
\Theta\left(\frac{Kp(1-\gamma)}{1+t}\right) \hspace{0.85cm} \text{ \emph{if} } Ip  = \Omega\left(\log\Lambda \right)\\
\Theta\left(\!\frac{Kp(1-\gamma)\log \Lambda}{(1+t)Ip\log \frac{\log \Lambda}{Ip}}\!\right) \text{ \emph{if} }  Ip\!\in\!\left[\! \Omega\left(\!\frac{1}{\polylog\!\Lambda}\!\right)\!,\!o(\!\log\!\Lambda\!)\!\right]\!.
\end{cases}
\end{align}

\end{theorem}
\begin{proof} 
The proof is deferred to Appendix~\ref{AP:lawtavg_iden}.
\end{proof} 



\section{Main Results: Data-Driven Approach}~\label{sec:results_data}
In this section, we will extend our analysis to the data-driven setting. Unlike in the previous section where we used a predetermined set of statistics $\mathbf{p}$, we will now exploit the users' content request histories to define the user activity levels as well as correlations.
To proceed with our analysis we need to define the time scales involved. In our setting, the entire time horizon is equal to the time it takes between two user-to-cache associations. This time horizon will be here subdivided into $S$ independent time slots, where one time slot corresponds to the amount of time that elapses from the appearance of one demand vector to the next demand vector. This dynamic time refinement captures the amount of memory of the system, and will capture how far back in history we can learn from regarding user activities.  

\begin{example}
In a scenario where users are assigned cache states once a week, then the time frame is equal to one week which is equal to 10080 minutes. In this same example, if we assume that independent demand vectors appear once every $10$ minutes, then the number of independent time slots $S$ is simply $S = \frac{10080}{10}= 1008$.
\end{example}

In our setting, users' requests are served simultaneously, starting at the very beginning of each time slot. Any content request received during a time slot is put on hold, to be served in the beginning of the next time slot. This justifies the use of the term \emph{dynamic duration} of each time slot $s\in [S]$, where this duration will be equal to the time needed to transmit all files that were requested during the previous time slot from the BS to the users.
We can now proceed with the details of our data-driven approach. 

Let $\mathbf{D} \in [0,1]^{S\times K}$ denote the \emph{user activity matrix}, of which the $(s,k)$ element $d_{s,k}$ is equal to $1$ if user $k$ requests content at time slot $s$, else $d_{s,k} = 0$. Then, for a given user-to-cache state association $\mathbf{G}$, the cache load vector for time slot $s \in [S]$ is denoted as $\mathbf{V}_s =\left[v_{s,1},  \dots, v_{s,\Lambda} \right]$, where $v_{s,\lambda}= \sum_{k=1}^K g_{\lambda,k}d_{s,k}$ is the number of active users at time slot $s$ that are storing the content of cache state $\lambda \in \left[\Lambda \right]$. The profile vector at time slot $s$ is denoted as $\mathbf{L}_s =\left[l_{s,1},  \dots, l_{s,\Lambda} \right]$, which is the sorted version of the cache load vector $\mathbf{V}_s$ in descending order. The average delay for a given user-to-cache state association $\mathbf{G}$ and a given user activity matrix $\mathbf{D}$, is given by
\begin{align} \label{eq:tavg_data}
  \overline{T}(\mathbf{G}) \triangleq \frac{1}{S}\sum_{s=1}^{S}  \sum_{\lambda=1}^{\Lambda-t}  l_{s,\lambda} \frac{{\Lambda-\lambda \choose t}}{{\Lambda \choose t}}.
\end{align}


Unlike in the statistical approach of Section~\ref{sec:results}, where an enormous number of possible profile vectors rendered the exact calculation of $\overline{T}(\mathbf{G})$ computationally intractable, in this current data-driven setting, the calculation of $\overline{T}(\mathbf{G})$ is direct even for large system parameters. This will allow us to design an algorithm that will find a user-to-cache association policy that is provably order-optimal.


An additional difference of the proposed data-driven problem formulation is that now this formulation inherits a crucial property of exploiting users' activity correlation in time. As previously discussed, users with similar request patterns will be associated with different cache states as this would guarantee more multicasting opportunities during the delivery phase. On the other hand, users that rarely request files at the same time, can be allocated the same cache state without any performance deterioration.


We now proceed to find an order-optimal user-to-cache state association $\hat{\mathbf{G}}$ corresponding to Problem~\ref{eq:opt_P1}. At this point we note that it is computationally intractable to brute-force solve Problem~\ref{eq:opt_P1} for large system parameters $K$, $\Lambda$ and $S$, since there are $\Lambda^{K}$ possible user-to-cache state associations, corresponding to an exhaustive-search computational complexity of $O\left(S\Lambda^{K+1}\right)$. Under these circumstances, the most common approach is to use computationally efficient algorithms to obtain an approximate solution that is away from the optimal solution within provable gaps. In the following subsection, we will present two such computationally efficient algorithms.

\subsection{Computationally efficient algorithms \& bounds on the performance}~\label{sec:algo_data}
We start with the following lemma which lower bounds the optimal average delay $\overline{T}^*$, optimized over all policies $\mathbf{G}$.
\begin{lemma}\label{le:tavg_op_lb}
The optimal average delay, optimized over all association policies, is lower bounded by
\begin{align} \label{eq:tavg_op_lb}
  &\overline{T}^* \geq \frac{1}{S} \sum_{s\in [S]} \left( \left\lfloor\frac{d_s}{\Lambda}\right\rfloor +1 \right)\frac{\Lambda-t}{1+t}  -\frac{1}{S}   \sum_{s\in \mathbf{S}_2}  \frac{ {\Lambda-A_s \choose t+1}}{{\Lambda \choose t}},  
\end{align}
where $d_s=\sum_{k\in [K]}d_{s,k}$, $A_s=d_s- \Lambda \left\lfloor \frac{d_s}{\Lambda}\right\rfloor$, and where $\mathbf{S}_2 \subseteq [S]$ is the set of time slots for which $A_s < \Lambda-t$.  
\end{lemma}
\begin{proof}
The proof is deferred to Appendix \ref{AP:tavg_op_lb}.
\end{proof}

The bound provided in Lemma \ref{le:tavg_op_lb} will serve as a benchmark for numerical performance evaluation of various user-to-cache state association algorithms.

We now proceed to present our computationally efficient algorithms. In the following, $\mathcal{G}$ will denote the set form of the user-to-cache state association matrix $\mathbf{G}$, where $(\lambda,k) \in \mathcal{G}$ if $g_{\lambda, k}=1$. Similarly, $\mathcal{G}^{(\lambda)}= \left\{k : (\lambda, k) \in \mathcal{G}\right\}$ will denote the set of users that are storing the content of cache state $\lambda \in [\Lambda]$. Note that there is a direct correspondence between $\mathbf{G}$ and $\mathcal{G}$, and the two terms can be used interchangeably.

\subsubsection{Algorithm \ref{alg:algo_1}:} Problem~\ref{eq:opt_P1} belongs to the family of well-known vector scheduling problems~\cite{chekuri_siam04, im_focs15}, whose aim is to optimally assign each of the  $S$-dimensional $K$ jobs (i.e., the $S$-dimensional $K$ vectors that are drawn from the columns of the user activity matrix $\mathbf{D}$) to one of the machines $\lambda \in [\Lambda]$ (i.e., cache states) with the objective of minimizing the maximum machine load (i.e., $\max_{s\in [S]} l_{s,1}$), or with the objective of minimizing the norm of the machine loads. One can see that the vector scheduling problem is the generalization of a classical load balancing problem, where each job has a vector load instead of a scalar load.


We adopt the vector scheduling algorithm of~\cite[Section II-B3]{im_focs15} to find the optimal user-to-cache state association within provable gaps. \textbf{Algorithm \ref{alg:algo_1}} consists of three parts. The first part is the data transformation, where the user activity matrix $\mathbf{D}$ is scaled according to step 00.  The second part (steps 01 to 08) is the deterministic user-to-cache state association, where for each user $k\in[K]$, we find the cache state $\hat{\lambda}\in [\Lambda]$ according to step 02. If the scaled load (cf. step 03) of cache $\hat{\lambda}$ after the assignment of user $k$ is less than $\frac{30\log S}{\log \log S}+1$ for all time slots $s\in [S]$, then user $k$ is assigned to cache state $\hat{\lambda}$. Otherwise user $k$ is not assigned to any of the cache states, and is instead added to a set of residual users denoted by $\mathcal{K}_r$, and will be associated to a cache state later in the third part of \textbf{Algorithm \ref{alg:algo_1}}. The outcome of the second part is the user-to-cache state association $\mathcal{G}_1$ for users in $[K]/\mathcal{K}_r$. Next, the third part (steps 09 to 13) completes the association of the residual users in $\mathcal{K}_r$. Each user $k\in \mathcal{K}_r$ is assigned to cache $\hat{\lambda}\in [\Lambda]$ according to step 11. The outcome of this part is the user-to-cache state association $\mathcal{G}_2$ for users in $\mathcal{K}_r$. The final user-to-cache state association strategy for all users is then given by $\mathcal{G} = \mathcal{G}_1\cup \mathcal{G}_2$. 

\begin{algorithm}
  \caption{\ }  \label{alg:algo_1}
  \begin{algorithmic}
 \STATE   \textbf{Input:} $\mathbf{D}$, $K$, $\Lambda$, and $S$ \\
\STATE \textbf{Output:} $\mathcal{G}$ \\
\STATE \textbf{Initialization:} $\mathcal{G}_1 \leftarrow \emptyset$; $\mathcal{G}_2 \leftarrow \emptyset$;  $\mathcal{K}_r \leftarrow \emptyset$; $\alpha=\frac{10\log S}{\log \log S}$ \\
\STATE Step 00: $\bar{d}_{s,k}  \leftarrow \min\left(\frac{\Lambda \ d_{s,k}}{\sum_{i\in[K]} d_{s,i}},1\right)$ \ \ $\forall \ s\in [S], k\in[K]$  
\STATE Step 01: \textbf{for} $k$ \textbf{from} 1 \textbf{to} $K$ \textbf{do} \\
\STATE Step 02: \ \ $\hat{\lambda}\! \leftarrow  \arg\!\!\min\limits_{\lambda \in [\Lambda]} \sum\limits_{s=1}^{S}\sum\limits_{\lambda=1}^{\Lambda} \left(\frac{1}{\alpha}\right)^{ \frac{\alpha}{\Lambda}\!\!\!\!\!\!\sum\limits_{i \in (\mathcal{G}_1\cup (\lambda, k))} \!\!\!\!\!\!\!\!\bar{d}_{s,i} -\!\!\!\!\sum \limits_{j \in (\mathcal{G}_1\cup (\lambda, k))^{\lambda}} \!\!\!\!\!\!\!\!\! \bar{d}_{s,j}}$ \\
\STATE Step 03: \ \ \textbf{if} $\sum\limits_{k \in (\mathcal{G}_1\cup (\hat{\lambda}, k))^{\hat{\lambda}}} \!\!\!\!\!\!\!\!\! \bar{d}_{s,k} < 3\alpha+1$ $\forall$ $s\in[S]$ \\
\STATE Step 04: \ \ \ \ $\mathcal{G}_1 \leftarrow \mathcal{G}_1\cup (\hat{\lambda},  k)$ \\
\STATE Step 05: \ \ \textbf{else}  \\
\STATE Step 06: \ \ \ \ $\mathcal{K}_r \leftarrow \mathcal{K}_r \cup k$ \\
\STATE Step 07: \ \ \textbf{end if} \\
\STATE Step 08: \textbf{end for} \\
\STATE Step 09: \textbf{for} $c$ \textbf{from} 1 \textbf{to} $\left|\mathcal{K}_r\right|$ \textbf{do} \\
\STATE Step 10: \ \ $k =\mathcal{K}_r(c)$ \\
\STATE Step 11: \ \ $\hat{\lambda}\! \leftarrow  \arg\!\!\min\limits_{\lambda \in [\Lambda]} \left( \max\limits_{s\in[S]}  \sum \limits_{j \in (\mathcal{G}_2\cup (\lambda, k))^{\lambda}} \!\!\!\!\!\!\!\!\! \bar{d}_{s,j} \right)$ \\
\STATE Step 12: \ \ $\mathcal{G}_2 \leftarrow \mathcal{G}_2\cup (\hat{\lambda},  k)$ \\
\STATE Step 13: \textbf{end for} \\
\STATE Step 14: $\mathcal{G} \leftarrow \mathcal{G}_1\cup \mathcal{G}_2$ \\
  \end{algorithmic}
\end{algorithm}

 \begin{theorem}\label{th:tavg_alg1_ub}
When there are at least $\Lambda$ requests at each time slot $s\in [S]$, the average delay $\overline{T}(\mathbf{G})$ corresponding to the user-to-cache state association $\mathbf{G}$ obtained from \textbf{Algorithm \ref{alg:algo_1}} is bounded by
\begin{align} \label{eq:tavg_alg1_ub}
  \overline{T}(\mathbf{G}) = O\left(\frac{\log S}{\log \log S}\overline{T}^*\right),
\end{align}
which proves that \textbf{Algorithm \ref{alg:algo_1}} is at most a factor $O\left(\frac{\log S}{\log \log S}\right)$ from the optimal.
\end{theorem}
\begin{proof}
The proof is deferred to Appendix \ref{AP:tavg_alg1_ub}.
\end{proof}

\begin{proposition}\label{prop:algo1}
The time complexity of \textbf{Algorithm \ref{alg:algo_1}} is $O(\Lambda^2KS)$.
\end{proposition}
\begin{proof}
The first part of \textbf{Algorithm \ref{alg:algo_1}} runs for $K$ iterations and in each iteration, the evaluation at step 02 takes at most $\Lambda^2S$ basic operations. Then, the second part of \textbf{Algorithm \ref{alg:algo_1}} runs for at most $K$ iterations and in each iteration, the evaluation at step 11 takes at most $\Lambda S$ basic operations. Thus the time complexity of \textbf{Algorithm \ref{alg:algo_1}} is $O(\Lambda^2KS)$. 
\end{proof}
Directly from above, we can see that \textbf{Algorithm \ref{alg:algo_1}} is significantly faster than the exhaustive search algorithm for which as we recall the time complexity  was $O\left(S\Lambda^{K+1}\right)$. 


\subsubsection{Algorithm \ref{alg:algo_2}:}
The main intuition behind \textbf{Algorithm \ref{alg:algo_2}} is to exploit the fact that both ${\Lambda-\lambda \choose t}$ and $l_{s,\lambda}$ are non-increasing with $\lambda$; a fact that directly follows from~\eqref{eq:tavg_data}. Thus, the optimal user-to-cache state association strategy is the one that minimizes the variances of the cache load vectors $\mathbf{V}_s$ over all time slots. \textbf{Algorithm \ref{alg:algo_2}} aims to heuristically minimize the sum of squares of cache populations over all time slots, which is equivalent to minimizing the sum of variances of the cache load vectors over all time slots. \textbf{Algorithm \ref{alg:algo_2}} works in $K$ iterations. At each iteration, it finds a pair of a user $\hat{k}$ and a cache state $\hat{\lambda}$ according to step 02 of \textbf{Algorithm \ref{alg:algo_2}} and assigns user $\hat{k}$ to cache state $\hat{\lambda}$.

\begin{algorithm}
  \caption{}
  \label{alg:algo_2}
  \begin{algorithmic}
 \STATE   \textbf{Input:} $\mathbf{D}$, $K$, and $\Lambda$ \\
\STATE \textbf{Output:} $\mathcal{G}$ \\
\STATE \textbf{Initialization:} $\mathcal{G} \leftarrow \emptyset$;  $\mathcal{K} \leftarrow [K]$ \\
\STATE Step 01: \textbf{for} $i$ \textbf{from} 1 \textbf{to} $K$ \textbf{do} \\
\STATE Step 02: \ \ $[\hat{\lambda},  \hat{k}]\!\! \leftarrow  \arg\!\!\!\!\!\!\min\limits_{\lambda \in [\Lambda], k \in \mathcal{K}  } \sum\limits_{s\in [S]}\sum\limits_{i\in [\Lambda]} \left(\sum\limits_{j\in (\mathcal{G}\cup (\lambda, k))^{(i)} }\!\!\!\!\!\! d_{s,j}\right)^2$ \\
\STATE Step 03: \ \ $\mathcal{G} \leftarrow \mathcal{G}\cup (\hat{\lambda},  \hat{k})$ \\
\STATE Step 04: \ \ $\mathcal{K} \leftarrow \mathcal{K} \symbol{92} \hat{k}$ \\
\STATE Step 05: \textbf{end for} \\
  \end{algorithmic}
\end{algorithm}

\begin{proposition}\label{prop:algo2}
The time complexity of \textbf{Algorithm \ref{alg:algo_2}} is $O(\Lambda^2K^2S)$.
\end{proposition}
\begin{proof}
\textbf{Algorithm \ref{alg:algo_2}} runs for $K$ iterations and in each iteration, the evaluation at step 02 takes at most $K\Lambda^2S$ basic operations. Thus the time complexity of \textbf{Algorithm \ref{alg:algo_2}} is $O(\Lambda^2K^2S)$. 
\end{proof}
We can see that the time complexity of \textbf{Algorithm \ref{alg:algo_2}} is $K$ times higher than the time complexity of \textbf{Algorithm \ref{alg:algo_1}}. However, in Section~\ref{sec:validation} we numerically show that \textbf{Algorithm \ref{alg:algo_2}} performs better than \textbf{Algorithm \ref{alg:algo_1}}.



\section{Numerical Validation}\label{sec:validation}
In this section, we numerically validate our analytical bounds, and evaluate the performance of the different proposed user-to-cache state association algorithms. 

\subsection{Statistical approach}
We first evaluate our proposed analytical bounds in Theorem~\ref{th:UBLB} and Theorem~\ref{th:UBLB_iden} for the statistical setting using the \emph{sampling-based numerical} (SBN) approximation method, where for any given $\mathbf{G}$, we generate a sufficiently large set $\mathcal{L}_1$ of randomly generated profile vectors $\mathbf{L}$ based on user activity vector $\mathbf{p}$ and where we subsequently approximate $ \overline{T}(\mathbf{G})$ as
\begin{align}\label{eq:tavgs} 
 \overline{T}(\mathbf{G})\approx \frac{1}{|\mathcal{L}_1|}\sum_{\mathbf{L}\in \mathcal{L}_1} T(\mathbf{L}),
 \end{align}
where $T(\mathbf{L})$ is defined in~\eqref{eq:TL}. 
For our evaluations involving an arbitrary user activity level vector $\textbf{p}$, we adopt the Pareto principle to generate the synthetic user activity level vector $\textbf{p}$. According to the Pareto principle, $80\%$ of consequences (content requests) come from $20\%$ of causes (users). To be exact, each user $k \in [K]$ has a request with probability
\begin{align}\label{eq:pkdata}
p_k =
    \begin{cases}
     \frac{1}{\sum_{i=1}^5 i^{-2.7}} \text{ \emph{if}} \ \ k  = [1, 2, \cdots, 0.2K] \\ 
     \frac{2^{-2.7}}{\sum_{i=1}^5 i^{-2.7}} \text{ \emph{if}} \ \ k  = [0.2K+1, 0.2K+2, \cdots, 0.4K] \\ 
     \frac{3^{-2.7}}{\sum_{i=1}^5 i^{-2.7}} \text{ \emph{if}} \ \ k  = [0.4K+1, 0.4K+2, \cdots, 0.6K] \\ 
     \frac{4^{-2.7}}{\sum_{i=1}^5 i^{-2.7}} \text{ \emph{if}} \ \ k  = [0.6K+1, 0.6K+2, \cdots, 0.8K] \\ 
     \frac{5^{-2.7}}{\sum_{i=1}^5 i^{-2.7}} \text{ \emph{if}} \ \ k  = [0.8K+1, 0.8K+2, \cdots, K]. \\ 
    \end{cases}
\end{align}
The intuition behind~\eqref{eq:pkdata} is that users are divided into $5$ equipopulated groups, and the users that belong to the same group have the same activity levels. 
The activity levels corresponding to these $5$ groups then follow the Power law with parameter $\alpha = 2.7$, and with these carefully selected parameters, the user activity pattern satisfies the Pareto principle (80/20 rule)~\cite{newman2005power}.


In Figure~\ref{fig:aub_vs_alb_arb}, we compare the analytical bounds in~\eqref{eq:UBtavg} and~\eqref{eq:LBtavg} for an arbitrary activity level vector $\mathbf{p}$, where this comparison uses the sampling-based numerical (SBN) approximation which is done for $|\mathcal{L}_1|=20000$ and random user-to-cache state association. Subsequently, Figure~\ref{fig:aub_vs_alb_iden} compares the analytical bounds in~\eqref{eq:UBtavg_iden} and~\eqref{eq:LBtavg_iden} for uniform user activity level, where again the comparison is with sampling-based numerical (SBN) approximation which is done for $|\mathcal{L}_1|=20000$ and uniform user-to-cache state association. Both figures reveal the proposed analytical bounds to be very tight, where in particular, analytical upper bounds are indeed very close to the exact performance.
\begin{figure}[t]
\centering
\includegraphics[ width=.9\linewidth]{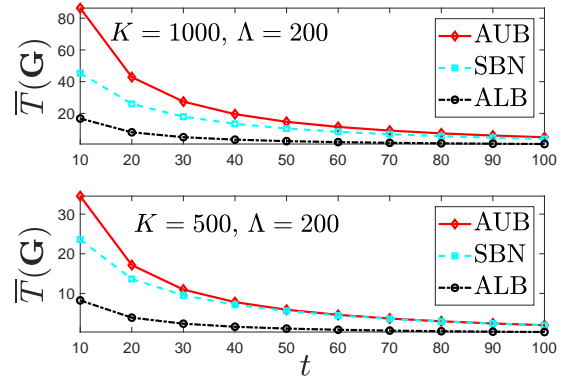} 
\caption{Analytical upper bound (AUB) from \eqref{eq:UBtavg} vs. analytical lower bound (ALB) from \eqref{eq:LBtavg} vs. sampling-based numerical (SBN) approximation in~\eqref{eq:tavgs} (for $|\mathcal{L}_1|=20000$, $\mathbf{p}$ in \eqref{eq:pkdata}, and random user-to-cache state association).}
\label{fig:aub_vs_alb_arb}
\end{figure}

\begin{figure}[t]
\centering
\includegraphics[ width=.9\linewidth]{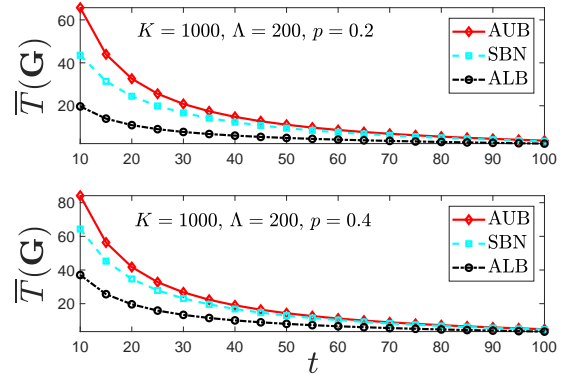} 
\caption{Analytical upper bound (AUB) from \eqref{eq:UBtavg_iden} vs. analytical lower bound (ALB) from \eqref{eq:LBtavg_iden} vs. sampling-based numerical (SBN) approximation in~\eqref{eq:tavgs} (for $|\mathcal{L}_1|=20000$ and uniform user-to-cache state association).}
\label{fig:aub_vs_alb_iden}
\end{figure}
\begin{figure}[t]
\centering
\includegraphics[width=.9\linewidth]{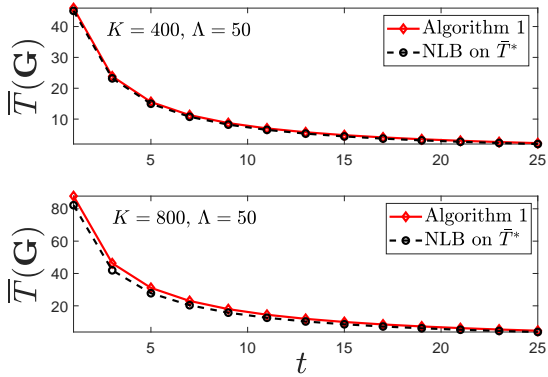} 
\caption{SBN from~\eqref{eq:tavgs} of \textbf{Algorithm \ref{alg:algo_p}}  vs. Numerical lower bound (NLB) on $\overline{T}^*$ from~\eqref{eq:tavg_op_lb} (for $|\mathcal{L}_1|=20000$ and $\mathbf{p}$ in \eqref{eq:pkdata} ).}
\label{fig:algo1_vs_lb}
\end{figure}

Next, we evaluate the performance of our first proposed user-to-cache state association algorithm (\textbf{Algorithm \ref{alg:algo_p}}) by comparing it with the numerical lower bound (NLB) on the delay $\overline{T}^*$ corresponding to the optimal user-to-cache state association $\hat{\mathbf{G}}$ of Lemma~\ref{le:tavg_op_lb}. Figure~\ref{fig:algo1_vs_lb} compares SBN approximation (once again done for $\vert \mathcal{L}_1 \vert = 20000$) for the user-to-cache state association obtained from \textbf{Algorithm \ref{alg:algo_p}} with the numerical lower bound (NLB) on $\overline{T}^*$ in~\eqref{eq:tavg_op_lb}. Again we observe that the performance corresponding to the user-to-cache state association $\mathbf{G}$ obtained from \textbf{Algorithm \ref{alg:algo_p}} is very close to NLB for $\overline{T}^*$.

\subsection{Data-driven approach}
For the data-driven approach, we synthetically generate a user activity matrix $\mathbf{D}$ following the Pareto principle. To be exact, we assume that user $k \in [K]$ develops a request (i.e., is active) with probability $p_k$ as in~\eqref{eq:pkdata} at each time slot $s \in [S]$. Then, for each time slot $s \in [S]$, we pick a random number $r_k$ between $0$ and $1$ for each user $k\in[K]$, and set $d_{s,k} = 1$ if $r_k\leq p_k$, and $d_{s,k} = 0$ if $r_k > p_k$, which yields a user activity matrix $\mathbf{D}$ satisfying the Pareto principle~\cite{newman2005power}. 

In Figure~\ref{fig:algo2_vs_algo3}, we compare the average delay $\overline{T}(\mathbf{G})$ in~\eqref{eq:tavg_data} corresponding to the user-to-cache state association obtained from \textbf{Algorithm \ref{alg:algo_1}} and \textbf{Algorithm \ref{alg:algo_2}} with the lower bound (LB) on $\overline{T}^*$ in~\eqref{eq:tavg_op_lb}. It turns out that both algorithms yield performances that are over close to the optimal LB on $\overline{T}^*$, with \textbf{Algorithm \ref{alg:algo_2}} having a slight advantage over \textbf{Algorithm \ref{alg:algo_1}}.

\begin{figure}[t]
\centering
\includegraphics[ width=.9\linewidth]{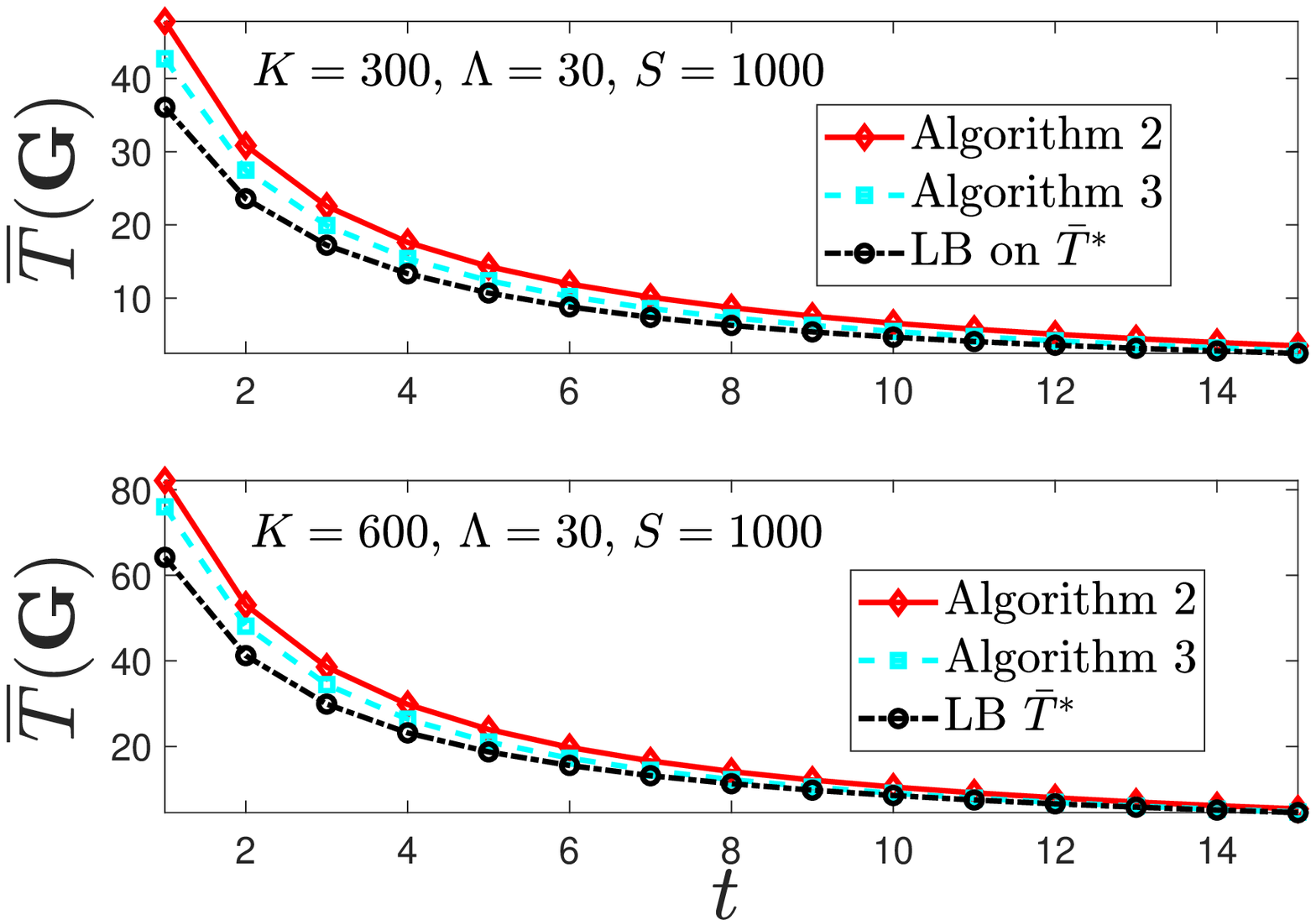} 
\caption{ $\overline{T}(\mathbf{G})$ of \textbf{Algorithm \ref{alg:algo_1}} and \textbf{Algorithm \ref{alg:algo_2}} from~\eqref{eq:tavg_data}  vs. lower bound (LB) on $\overline{T}^*$ from~\eqref{eq:tavg_op_lb}.}
\label{fig:algo2_vs_algo3}
\end{figure}
\begin{figure}[t]
\centering
\includegraphics[ width=.9\linewidth]{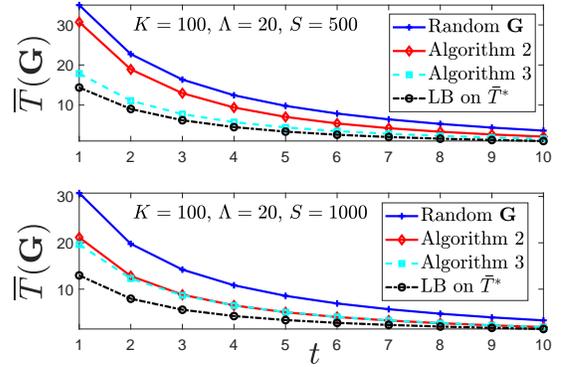} 
\caption{$\overline{T}(\mathbf{G})$ of random user-to-cache state association, the user-to-cache state associations obtained from \textbf{Algorithm \ref{alg:algo_1}}, and \textbf{Algorithm \ref{alg:algo_2}} from~\eqref{eq:tavg_data}, and the lower bound (LB) of~\eqref{eq:tavg_op_lb}.}
\label{fig:rand_algo2_vs_algo3}
\end{figure}

In our final evaluation, we highlight the importance of exploiting the user-activity patterns and finding an efficient user-to-cache state association. Figure~\ref{fig:rand_algo2_vs_algo3} compares $\overline{T}(\mathbf{G})$ values for random user-to-cache state association, and the user-to-cache state associations obtained from \textbf{Algorithm \ref{alg:algo_1}} and \textbf{Algorithm \ref{alg:algo_2}}, where the lower bound (LB) of~\eqref{eq:tavg_op_lb} serves as a benchmark. It turns out that both algorithms outperform random user-to-cache state association, and the corresponding delay performances perform very close to the optimal LB on $\overline{T}^*$, with once again \textbf{Algorithm \ref{alg:algo_2}} having a slight advantage over \textbf{Algorithm \ref{alg:algo_1}}.

\section{Conclusions}\label{sec:conclusions}
In this work we analyzed coded caching networks with finite number of cache states and a user-to-cache state association subject to a grouping strategy in the presence of heterogeneous user activity. Even though coded caching techniques rely on the assumption of having enough number of users to provide its theoretically promised gains, all the earlier works ignored the fact of heterogeneity in user activities, which in our opinion has direct practical ramifications, as it captures practical wireless networks more accurately.

We first presented a statistical analysis of the average worst-case delivery performance of state-constrained coded caching networks, and provided bounds and scaling laws under the assumption of probabilistic user-activity levels. We also proposed a heuristic user-to-cache state association algorithm with the ultimate goal of minimizing the average delay. 

Next, we extended our analysis to the data-driven setting, where we were able to learn from the past $S$ different demand vectors in designing the caching policy. By exploiting this bounded-depth user request history, the emphasis then was placed on finding the optimal user-to-cache state association, as computing the average delay for any given data is trivial. We proposed two algorithms for finding the optimal user-to-cache state association strategy, with the first algorithm providing the optimal within a constant gap, and with the second algorithm numerically outperforming the first one. 

For both aforementioned settings, the results highlighted the essence of exploiting the user activity level, and the importance of carefully associating users to cache states based on their activity patterns.

\appendix
\subsection{Proof of Theorem \ref{th:UBLB}} \label{AP:UBLB}
Exploiting the fact that in \eqref{eq:tavg}, both ${\Lambda-\lambda \choose t}$ and $E[l_{\lambda}]$ are non-increasing with $\lambda$, the average delay $\overline{T}(\mathbf{G})$ is bounded by
\begin{align}\label{eq:TUB}
\!\!\overline{T}(\mathbf{G}) & \leq \!   \sum_{\lambda=1}^{ \Lambda-t} E[ l_{1} ] \frac{{\Lambda-\lambda \choose t} }{{\Lambda \choose t}}   \stackrel{(a)}{=}   E[ l_{1} ] \frac{  {\Lambda \choose t+1}}{{\Lambda \choose t}} =   E[ l_{1} ] \frac{\Lambda-t}{1+t},
 \end{align}
 and
 \begin{align} \label{eq:TLB}
\overline{T}(\mathbf{G})& \stackrel{(b)}{\geq}  \frac{  E[ l_{1} ]   {\Lambda-1 \choose t} \!+ \! \sum_{\lambda=2}^{\Lambda-t} \frac{K_{\mathbf{p}}-E[ l_{1}]}{\Lambda-1}  {\Lambda-\lambda \choose t} }{{\Lambda \choose t}}\nonumber\\ &\stackrel{(c)}{=}   E[ l_{1} ]  \frac{   {\Lambda-1 \choose t}} {{\Lambda \choose t}}  + \frac{K_{\mathbf{p}}-E[ l_{1}]}{\Lambda-1}  \frac{  {\Lambda-1 \choose t+1}}{{\Lambda \choose t}} \nonumber \\
& =  \frac{\Lambda-t}{1+t} \left( \frac{E[ l_{1} ] t  }{ \Lambda-1} +\frac{K_{\mathbf{p}}}{\Lambda}  \frac{  \Lambda-t-1}{\Lambda-1} \right),
 \end{align}
 where in steps (a) and (c), we inherit the the column-sum property of Pascal's triangle yielding $\sum^{n}_{k=0} {k\choose t}= {n+1\choose t+1}$, while in step (b), we have\footnote{It straightforward to see that  $\sum\limits_{\lambda\in [\Lambda]}E[ l_{\lambda}]=    \sum\limits_{i=0}^K \sum\limits_{ \mathbf{L}\in \mathcal{L} : i=\sum\limits_{j \in [\Lambda]} l_{j}  } \sum\limits_{\lambda\in[\Lambda]} l_{\lambda} P(\mathbf{L}) = \sum\limits_{i=0}^K \sum\limits_{ \mathbf{L}\in \mathcal{L} : i=\sum\limits_{j \in [\Lambda]} l_{j}  }i P(\mathbf{L}) =  K_{\mathbf{p}}$.} $K_{\mathbf{p}}=\sum_{\lambda=1}^{\Lambda}E[ l_{\lambda}]$, and the fact that uniformity in $\mathbf{L}$ leads to the minimum $\overline{T}(\mathbf{G})$.

 Next, to complete the proof, we proceed to derive the expected number of active users that are storing the content of the most loaded cache state (i.e., $E[l_1]$), which is given by
   \begin{align}
E[l_{1}] &= \sum_{x=0}^{A-1} P[l_{1} > x]  = \sum_{x=0}^{A-1} \left(1- P[l_{1} \leq x] \right),
\end{align}
where $A=\max\left(\left\{|\mathbf{G}_{\lambda}|\right\}_{\lambda=1}^{\Lambda}\right)$, $\mathbf{G}_{\lambda}$ is the set of users caching the content of cache state $\lambda$, and  $P[l_{1} \leq x]$ is the probability that number of active users storing the content of the most loaded cache state are less than or equal to $x$. From \cite[Proposition 3]{caraux_92}, we have 
\begin{align}
  P[l_{1} \leq x] \geq \max \left(0,1-\Lambda+ \sum_{\lambda=1}^{\Lambda} F_{v_{\lambda}}(x)  \right)
\end{align}
and 
\begin{align}
P[l_{1} \leq x] \leq  \frac{\sum_{\lambda=1}^{\Lambda} F_{v_{\lambda}}(x) }{\Lambda}, 
\end{align}
 where $F_{v_{\lambda}}(x)$ is the probability that no more than $x$ users that are caching the content of cache state $\lambda \in [\Lambda]$ are active (i.e., $P[v_{\lambda} \leq x ]$). Then $E[l_{1}]$ is bounded by  
\begin{align}\label{eq:el1_ub}
E[l_{1}]&\leq A- \sum_{x=0}^{A-1} \max \left(0,1-\Lambda+ \sum_{\lambda=1}^{\Lambda} F_{v_{\lambda}}(x)  \right) 
\end{align}
and
\begin{align}\label{eq:el1_lb}
E[l_{1}]\geq A- \sum_{x=0}^{A-1} \frac{\sum_{\lambda=1}^{\Lambda} F_{v_{\lambda}}(x) }{\Lambda}.
\end{align}
For each cache state $\lambda \in [\Lambda]$, the corresponding random variable $v_{\lambda}$ follows the Poisson binomial distribution. Using Hoeffding's inequalities \cite[Theorem 2.1]{tang_arxiv19}, $F_{v_{i}}(x)$ is bounded by 
\begin{align}\label{eq:cdflb}
F_{v_{\lambda}}(x) \geq 
\begin{cases} 0  \hspace{2.5cm} \text{ for } 0 \leq x\leq \mu_{\lambda}-1    \\
F_{bin}\left(|\mathbf{G}_{\lambda}|, \frac{\mu_{\lambda}}{|\mathbf{G}_{\lambda}|}, x\right)\text{ for }  \mu_{\lambda} \leq x\leq |\mathbf{G}_{\lambda}|
\end{cases}
\end{align}
and
 \begin{align}\label{eq:cdfub}
F_{v_{\lambda}}(x) \leq 
\begin{cases}
F_{bin}\left(|\mathbf{G}_{\lambda}|, \frac{\mu_{\lambda}}{|\mathbf{G}_{\lambda}|}, x\right) \text{ for }  0 \leq x\leq \mu_{\lambda} -1 \\
1. \hspace{2.4cm}\text{ for }  x> \mu_{\lambda}-1,
\end{cases}
\end{align}
where $\mu_\lambda=\sum_{k\in \mathbf{G}_{\lambda}} p_{k}$ is the expected number active users that are storing the content of cache state $\lambda \in [\Lambda]$ and $F_{bin}\left(n, q, x\right) = \sum_{i=0}^{x} {n \choose i} q^i \left(1-q\right)^{n-i}$ is the Binomial cumulative distribution function.

Finally, the upper bound in \eqref{eq:UBtavg} can be obtained from \eqref{eq:TUB}, \eqref{eq:el1_ub}, and \eqref{eq:cdflb}; and the lower bound in \eqref{eq:LBtavg} can be obtained from \eqref{eq:TLB}, \eqref{eq:el1_lb}, and \eqref{eq:cdfub}.

\subsection{Proof of Theorem \ref{th:lawtavg}} \label{AP:lawtavg}
From~\eqref{eq:TUB} and~\eqref{eq:TLB}, we have, 
\begin{align}\label{eq:t_ub_law}
\!\!\overline{T}(\mathbf{G})  = \!  O\left( E[ l_{1} ] \frac{\Lambda-t}{1+t}\right),
 \end{align}
 and
 \begin{align} \label{eq:t_lb_law}
\overline{T}(\mathbf{G}) = \Omega\left( \frac{\Lambda-t}{1+t} \frac{E[ l_{1} ] t  }{ \Lambda-1}\right) .
 \end{align}
As $\frac{t}{\Lambda-1} \approx \gamma$ is a constant, we get the exact scaling law of $\overline{T}(\mathbf{G})$, which is given by
\begin{align}\label{eq:t_law}
\overline{T}(\mathbf{G}) = \Theta\left( E[ l_{1} ] \frac{\Lambda-t}{1+t}\right).
\end{align}
We know from~\cite[Proposition 1]{gascuel_92} that $E[l_1]$ is bounded by 
\begin{align}
\frac{1}{\Lambda}\sum_{\lambda=1}^{\Lambda} \mu_{\lambda}  &\leq E[l_{1}] \leq \frac{1}{\Lambda}\sum_{\lambda=1}^{\Lambda} \mu_{\lambda} \nonumber\\ &+\sqrt{\frac{\Lambda\!-\!1}{\Lambda}\sum_{\lambda=1}^{\Lambda} \!\left(\!\sigma_{\lambda}^2\!+\! \left(\!\mu_{\lambda}\!-\!\frac{1}{\Lambda}\sum_{\lambda=1}^{\Lambda} \mu_{\lambda}\!\right)^2\right) },
\end{align}
where $\mu_{\lambda}= \sum_{k\in G_{\lambda}} p_{k}$ and $\sigma^2_{\lambda}= \sum_{k\in G_{\lambda}} p_{k} (1-p_{k})$ are the mean and the variance of the number of active users that are caching the content of cache state $\lambda \in [\Lambda]$ respectively. After defining a new parameter $\mu= \frac{1}{\Lambda}\sum_{\lambda=1}^{\Lambda} \mu_{\lambda}$, we have 
\begin{align}
\!\!\overline{T}(\mathbf{G})  = \!  O\left( \left(\mu  + \sqrt{\sum_{i=1}^{\Lambda} [\sigma_i^2+ (\mu_i-\mu)^2] }\right) \frac{\Lambda-t}{1+t}\right),
 \end{align}
 and
 \begin{align}
\overline{T}(\mathbf{G}) = \Omega\left(  \mu \frac{\Lambda-t}{1+t}\right).
 \end{align}
\textcolor{black}{This concludes the proof of Theorem \ref{th:lawtavg}.}

\subsection{Proof of Theorem \ref{th:UBLB_iden}} \label{AP:UBLB_iden}
We start our proof by deriving the expected number of active users that are storing the content of the most loaded cache state (i.e., $E[l_1]$).  Assuming that each user independently requests a content with probability $p$, the probability that no more than $x$ out of $I$ users are active and storing the content of cache state $\lambda \in [\Lambda]$ is given by
\begin{align}
F_{v_{\lambda}}(x) =\sum_{i=0}^{x} {I \choose i} p^i \left(1-p\right)^{I-i}.
\end{align}
Then, the probability that $l_1$ (i.e., $\max(\mathbf{V})$, the maximum number of active users among all caches) is less than or equal to $j$, is equal to the probability of the event $v_{\lambda} \leq j$, $\forall$ $\lambda \in [\Lambda]$, and given by
\begin{align}
P[l_1 \leq x] = \prod_{\lambda=1}^{\Lambda} F_{v_{\lambda}}(x)= \left(\sum_{i=0}^{x} {I \choose i} p^i \left(1-p\right)^{I-i}\right)^{\Lambda}.
 \end{align}
Now, we can characterize $E[l_1]$ as follows
\begin{align}\label{eq:el1_iden}
E[l_1]\!&=\!\sum_{x=0}^{I-1} \left(1\!-\!P[l_1 \leq x] \right)\!=\!I\!-\! \sum_{x=0}^{I-1}\!\left(\!\sum_{i=0}^{x}\!{I \choose i}\!p^i\! \left(\!1\!-\!p\!\right)^{I\!-\!i}\right)^{\Lambda}.
\end{align}

Finally, we obtain the upper and lower bounds in Theorem \ref{th:UBLB_iden} by combining \eqref{eq:TUB} and \eqref{eq:TLB} with \eqref{eq:el1_iden}, respectively. 

\subsection{Proof of Theorem \ref{th:lawtavg_iden}} \label{AP:lawtavg_iden}
To prove Theorem~\ref{th:lawtavg_iden}, we will follow a similar approach as in~\cite{martin_98}. For each cache state $\lambda \in [\Lambda]$, we denote $Y_{\lambda}$ to be an indicator random variable, which is equal to $1$ if $v_{\lambda} \geq k_{\alpha}$, and it is equal to $0$ otherwise. It immediately follows that $E[Y_{\lambda}]= P[v_{\lambda} \geq k_{\alpha}]$, $\forall \lambda \in [\Lambda]$. Let $Y=\sum_{\lambda=1}^{\Lambda} Y_{\lambda}$ be the sum of the indicators over all cache states. Then, we have
\begin{align}\label{eq:EY}
    E[Y]=E\left[\sum_{\lambda=1}^{\Lambda} Y_{\lambda}\right]= \sum_{\lambda=1}^{\Lambda}E\left[ Y_{\lambda}\right] = \Lambda  P[v_{\lambda} \geq k_{\alpha}].
\end{align}
From \cite[Section 2]{martin_98}, we inherit the following properties that are drawn from the outcomes of Markov's inequality and Chebyshev's inequality
\begin{align} \label{eq:p_Y0}
P[Y=0]= 
\begin{cases} 
1-o(1)  \hspace{0.45cm} \text{ if } \log(E[Y])\rightarrow -\infty \\
o(1)  \hspace{1cm} \text{ if } \log(E[Y]) \rightarrow \infty. \\
\end{cases}
\end{align}
Consequently, the probability that there exists at least one cache state $\lambda$ for which the number of active users $v_{\lambda}$ is at least $k_{\alpha}$ is given by 
\begin{align} \label{eq:p_Y1}
P[Y\geq 1]= 
\begin{cases} 
o(1)  \hspace{1cm} \text{ if } \log(E[Y])\rightarrow -\infty \\
1-o(1)  \hspace{0.45cm} \text{ if } \log(E[Y]) \rightarrow \infty. \\
\end{cases}
\end{align}
We now proceed with the following results which are crucial for the derivation of the asymptotics of $E[l_1]$. 
\begin{lemma}[\protect{\cite[Lemma 2]{martin_98} - adaptation}]\label{le:cdf_V}
For a positive constant $c$, if $Ip+1 \leq x\leq \left(\log I\right)^c$,  then 
\begin{align}\label{eq:cdf_V_c1}
   P[v_{\lambda} \geq x]=e^{x(\log Ip-\log x + 1)-Ip+O\left(log^{(2)}I\right)}
\end{align}
and if $x=Ip + o\left( (p(1-p)I)^{\frac{2}{3}}\right)$ and $z=\frac{x-Ip}{\sqrt{p(1-p)I}}$ tends to infinity, then 
\begin{align}\label{eq:cdf_V_c2}
   P[v_{\lambda} \geq x]= e^{-\frac{z^2}{2}-\log z-\log \sqrt{2\pi}+ o(1)}
\end{align}
\end{lemma}
\begin{proof}
The result comes directly from \cite[Lemma 2]{martin_98}.
\end{proof}
\begin{lemma}\label{le:cdf_l1}
In a $K$-user $\Lambda$-cache state setting where each user requests a content with probability $p$, the probability that the maximum number of active users among all caches is less than or equal to $k_{\alpha}$, takes the form
\begin{align} \label{eq:cdf_l1}
P[l_1 \geq k_{\alpha}]= 
\begin{cases} 
 o\left( 1\right) \hspace{0.55cm}\text{ \emph{if} } \ \   \alpha  > 1 \\
 1- o\left( 1\right) \text{ \emph{if} }    \ \ 0 < \alpha < 1,
\end{cases}
\end{align} 
for 
\begin{align} \label{eq:ka}
k_{\alpha}= 
\begin{cases} 
Ip+ \sqrt{2\alpha Ip(1-p)\log(\Lambda)}, \text{ \emph{if} }  Ip  = \omega\left(\left(\log\Lambda\right)^3\! \right)\\
\left(\!1\!+\!\alpha\sqrt{\frac{2\log \Lambda}{Ip}}\right)\!Ip, \text{ \emph{if} } Ip\!\in\!\left[\omega\!\left(\!\log\!\Lambda\!\right),\!O(\!\polylog\! \Lambda\!)\right]\\
 \left(\alpha+ e-1\right)Ip, \text{ \emph{if} }  Ip = \Theta \left(\log \Lambda\right) \\
\frac{\log \Lambda}{\log\!\frac{\log \Lambda}{Ip}}\!\left(\!1\!+\!\alpha\!\frac{\log^{(2)}\!\frac{\log \Lambda}{Ip}}{\log\!\frac{\log \Lambda}{Ip}}\!\right)\!,\! \text{ \emph{if} } Ip\! \in\!\left[\!\Omega\!\left(\!\frac{1}{\polylog\!\Lambda}\!\right)\!,\!  o(\!\log\!\Lambda\!)\!\right]\!.
\end{cases}
\end{align}
\end{lemma}
\!\!\begin{proof}
We begin the proof for the case of $Ip\! =\! \omega \left(\left(\log \Lambda\right)^3\right)$. Let $k_{\alpha}\! =\! Ip+ \sqrt{2\alpha Ip(1\!-\!p)\log \Lambda}$, then from \eqref{eq:EY} and \eqref{eq:cdf_V_c2}, we have
\begin{align}
&\log(E[Y])=  \log \Lambda -\frac{z^2}{2}-\log z-\log \sqrt{2\pi} +o(1) \nonumber \\
 & =\log\Lambda\left(\!1 \!-\!\alpha-\!\frac{\log\!2\alpha+\log^{(2) 
 }\Lambda}{2\log\Lambda}\right)\!-\!\log\sqrt{2\pi}\!+\!o(1).
\end{align}
Using \eqref{eq:p_Y1}, we conclude the proof for this case as for $\Lambda \rightarrow \infty$, we have
\begin{align}
\log(E[Y]) \longrightarrow  \begin{cases}\!
-\infty \hspace{1.2cm}\text{ if }  \alpha>1 \\
\infty \hspace{1.35cm}\text{ if }  0 < \alpha < 1.
\end{cases}
\end{align}
Next, we proceed with the case of $Ip \in\left[ \omega\left(\log\Lambda\right),  O( \polylog(\Lambda))\right]$. We first define $g \triangleq O\left( \polylog(\Lambda)\right)$. Then, assuming that $k_{\alpha}=\left(1+\alpha\sqrt{\frac{2}{g}}\right)Ip$ and $Ip= g\log(\Lambda)$, from \eqref{eq:EY} and \eqref{eq:cdf_V_c1}, we have
\begin{align}
&\log(E[Y])\!= \!\log \Lambda \!+\!k_{\alpha}\left(\log Ip\!-\!\log k_{\alpha} \!+\!1\right)\!-\!Ip\!+\!O\left(\!\log^{(2)}\!I\!\right) \nonumber\\
&= \log\Lambda -k_{\alpha}\log\left(1+\alpha\sqrt{\frac{2}{g}}\right)+k_{\alpha}-Ip+O\left(\log^{(2)}I\right) \nonumber\\
&\stackrel{(a)}{=}\!\log\!\Lambda\! -\!k_{\alpha}\alpha\sqrt{\frac{2}{g}}\!\left(\!1-\!\alpha\!\sqrt{\frac{1}{2g}}\!+\! o\!\left(\!\alpha\sqrt{\frac{2}{g}}\!\right)\!\right)\!+\alpha Ip\sqrt{\frac{2}{g}}\nonumber\\
&+O\!\left(\!\log^{(2)}\!I\!\right) \nonumber\\
&=\!\log\!\Lambda\!\left(\!1\!-\!\alpha^2(1 \!+ \!o\left(1\right)) \!+\!\left(1\!-\! o\left(1\right)\right)\alpha^3\!\sqrt{\frac{2}{g}}\!+O\!\left(\!\frac{\log^{(2)}\!I}{\log \Lambda}\!\right)\!\right)\!,
\end{align}
where in step (a), we used the Maclaurin series expansion of the logarithm function, i.e., $\log(1+x)=x-0.5x^2+o(x^2)$. 
Using \eqref{eq:p_Y1}, we conclude the proof for this case as for $\Lambda \rightarrow \infty$, $\log(E[Y])$ converges to $(1\!-\!\alpha^2) \log \Lambda $, and we obtain
\begin{align}
\log(E[Y]) \longrightarrow  \begin{cases}\!
-\infty \hspace{1.2cm}\text{ if }  \alpha>1 \\
\infty \hspace{1.35cm}\text{ if }  0 < \alpha < 1.
\end{cases}
\end{align}
Now, we proceed with the case of $Ip =  \Theta\left(\log\Lambda\right)$. Assuming that $k_{\alpha} = \left(\alpha+ e-1\right)Ip$ and $Ip= \log \Lambda$, from \eqref{eq:EY} and \eqref{eq:cdf_V_c1}, we obtain
\begin{align}
&\log(E[Y])\!= \!\log\!\Lambda\!+\!k_{\alpha}\!\left(\log (Ip)\!-\!\log k_{\alpha} \!+\!1\!\right)\!-\!Ip\!+\!O\!\left(\log^{(2)}\!I\!\right) \nonumber\\
&=k_{\alpha}\left(1\!-\!\log\left(\alpha+ e-1\right)\right) +O\left(log^{(2)}I\right)\nonumber\\
&=\log\Lambda \left(\left(\alpha\!+\!e\!-1\right)\left(1\!-\!\log\left(\alpha\!+\! e\!-\!1\right)\right)+O\left(\frac{log^{(3)}\Lambda}{p\log \Lambda}\right)\right).
\end{align}
Using \eqref{eq:p_Y1}, we conclude the proof for this case as for $\Lambda \rightarrow \infty$, we have
\begin{align}
\log(E[Y]) \longrightarrow  \begin{cases}
-\infty \hspace{1.15cm}\text{ if }  \alpha>1 \\
\infty \hspace{1.35cm}\text{ if }  0 < \alpha < 1 
\end{cases}
\end{align}
Finally, we consider the case of $Ip \in\left[ \Omega\left(\frac{1}{\polylog\Lambda}\right),  o( \log(\Lambda))\right]$. We first define $g \triangleq O\left( \polylog(\Lambda)\right)$. Then, assuming that $k_{\alpha} = \frac{\log \Lambda}{\log g}\left(1+\alpha\frac{\log^{(2) }g}{\log g}\right)$ and $Ip=\frac{\log(\Lambda)}{g}$, from \eqref{eq:EY} and \eqref{eq:cdf_V_c1}, we obtain
\begin{align}
&\log(E[Y])\!= \!\log\!\Lambda \!+\!k_{\alpha}\!\left(\log(Ip)\!-\!\log k_{\alpha}\!+\!1\!\right)\!-Ip\!+\!O\!\left(\!\log^{(2)}\!I\!\right) \nonumber\\
&=\!\log \Lambda \!+\!k_{\alpha}\Bigg(\!\!\log^{(2)}\! \Lambda\!-\! \log g \!-\!\log^{(2)} \!\Lambda \!+\!  \log^{(2)}\! g\!\nonumber\\
&-\! \log \left(\!1\!+\!\alpha\frac{\log^{(2) }g}{\log g}\!\right)\!+\!1\!\!\Bigg)\!-\!\frac{\log(\Lambda)}{g}\!+\!O\left(\!\log^{(2)}\!I\!\right) \nonumber\\
&\stackrel{(a)}{=}\log\!\Lambda\! +\!k_{\alpha}\Bigg(\!\!1\! -\!\log\!g \!+\!  \log^{(2)} g\!-\! \bigg[\! \alpha\frac{\log^{(2)}\!g}{\log g} \!-\!\frac{1}{2}\!\left(\!\alpha\frac{\log^{(2)}\!g}{\log g}\!\right)^2 \nonumber\\&+ o\bigg(\bigg(\alpha\frac{\log^{(2) }g}{\log g}\bigg)^2\bigg)\bigg]\Bigg) -\frac{\log(\Lambda)}{g}+O\left(\log^{(2)}I\right) \nonumber\\
&=\frac{\log \Lambda\log^{(2) }g}{\log g}\left(1 -\alpha + \alpha\frac{\log^{(2) }g}{\log g} +\frac{1}{\log^{(2) }g}   -\frac{\log g}{g\log^{(2) }g}\right.\nonumber\\&\left.+O\left(\frac{\log^{(2)}I \log g}{\log\Lambda\log^{(2)}g}\right) +\alpha^2\frac{\log^{(2) }g}{(\log g)^2}\Bigg[\!-0.5\!+\!0.5\!\left(\alpha\frac{\log^{(2) }g}{\log g}\right)\right. \nonumber\\&- \left.o\left(\left(\alpha\frac{\log^{(2) }g}{\log g}\right)\right)-o(1)\right)\Bigg],
\end{align}
where in step (a), we used the Maclaurin series expansion of the logarithm function, i.e., $\log(1+x)=x-0.5x^2+o(x^2)$. Using \eqref{eq:p_Y1}, we conclude the proof for this case as for $\Lambda \rightarrow \infty$, $\log(E[Y])$ converges to $\frac{\log \Lambda\log^{(2) }g}{\log g} \left( 1   -\alpha\right)$, and we obtain
\begin{align}
\log(E[Y]) \longrightarrow  \begin{cases} -\infty \hspace{1.15cm}\text{ if }  \alpha>1 \\
\infty \hspace{1.35cm}\text{ if }  0 < \alpha < 1
\end{cases}
\end{align}
This concludes the proof of Lemma~\ref{le:cdf_l1}.
\end{proof}
With Lemma~\ref{le:cdf_l1} at hand, we proceed to characterize $E[l_1]$. Let us first consider the case of $\alpha > 1$, for which we have 
\begin{align} \label{eq:ubl1}
&E[l_1] = \sum_{j=1}^{k_{\alpha}-1} P[l_1 \geq j] +  P[l_1 \geq k_{\alpha} ] +\sum_{j=k_{\alpha}+1}^{I} P[l_1 \geq j]\nonumber\\
&\stackrel{(a)}{\leq} k_{\alpha}-1 +  o(1) + (I-k_{\alpha})o(1)  = O\left( k_{\alpha}\right), 
\end{align}
where in step (a), we use the fact that $P[l_1 \geq j]$ is at most $1$ for $j= \left[ 1, \cdots k_{\alpha}-1\right]$, and if $P[l_1 \geq k_{\alpha}] =o(1)$ then $P[l_1 \geq j]$ is at most $o(1)$  for $j= \left[ k_{\alpha}+1, \cdots I\right]$. 

Similarly, for $ 0 < \alpha < 1$, we have 
\begin{align}\label{eq:lbl1}
E[l_1] &= \sum_{j=1}^{k_{\alpha}-1} P[l_1 \geq j] +  P[l_1 > k_{\alpha} ] +\sum_{j=k_{\alpha}+1}^{I} P[l_1 \geq j] \nonumber \\
&\stackrel{(a)}{\geq} (k_{\alpha}-1) (1-o(1)) + 1- o(1) = \Omega\left( k_{\alpha}\right), 
\end{align}
where in step (a), we use the fact that $\sum_{j=k_{\alpha}+1}^{I} P[l_1 \geq j] \geq 0$, and if $P[l_1 \geq k_{\alpha}] =1-o(1)$ then $P[l_1 \geq j]$ is at least $1-o(1)$ for $j= \left[1, \cdots, k_{\alpha}-1\right]$. Combining~\eqref{eq:ka},~\eqref{eq:ubl1}, and~\eqref{eq:lbl1}, we have
\begin{align}\label{eq:AL1}
E[l_1] =
 \begin{cases} 
\Theta\left(Ip+ \sqrt{ Ip(1-p)\log(\Lambda)}\right), \text{ if }  Ip\!=\!\omega\!\left(\!\left(\!\log\!\Lambda\!\right)^3\!\right)\\
\Theta\left(\!Ip\!+\!\sqrt{Ip\log\!\Lambda}\right)\!, \text{ if } Ip\!\in\!\left[ \Omega\!\left(\!\log\!\Lambda\!\right),  O(\!\polylog\!\Lambda\!)\right]\\
\Theta \left(\frac{\log \Lambda}{\log \frac{\log \Lambda}{Ip}}\right), \text{ if }  Ip \in\left[\Omega\left(\frac{1}{\polylog\!\Lambda}\right),  o(\log\Lambda)\right].
\end{cases}
\end{align}
From~\eqref{eq:UBtavg_iden} and~\eqref{eq:LBtavg_iden}, we have, 
\begin{align}
\overline{T}(\mathbf{G}) = O\left( \frac{Kp(1-\gamma)}{1+t} \frac{E[l_1]}{Ip}\right) \end{align}
and 
 \begin{align}
 \overline{T}(\mathbf{G}) =\Omega\left( \frac{Kp(1-\gamma)}{1+t}\left(\frac{ E[l_{1}] t  }{ Ip(\Lambda-1)} \right)\right).
 \end{align}
As $\frac{t}{\Lambda-1} \approx \gamma$ is a constant, we get the exact scaling law of $\overline{T}(\mathbf{G})$, which is given by
\begin{align}\label{eq:t_ex_law}
\overline{T}(\mathbf{G}) = \Theta\left( \frac{Kp(1-\gamma)}{1+t} \frac{E[l_1]}{Ip}\right) 
\end{align}
Combining~\eqref{eq:t_ex_law} with~\eqref{eq:AL1}, we obtain
\begin{align}
 \overline{T}\!(\!\mathbf{G}\!)\!=\!
 \begin{cases} 
\Theta\!\left(\!\frac{Kp(1-\gamma)}{1+t}\! \left(\!1\!+\!\sqrt{ \frac{(\!1\!-\!p\!)\!\log\!\Lambda}{Ip}}\right)\!\right)\!, \text{ if }  Ip\!=\! \omega\!\left(\!\left(\!\log\!\Lambda\!\right)^3\!\right)\\
\Theta\!\left(\!\frac{K\!p\!(\!1\!-\!\gamma\!)}{1\!+\!t}\!\!\left(\!1\!\!+\!\!\sqrt{\frac{\log\!\Lambda}{I\!p}}\!\right)\!\right)\!, \text{ if } I\!p\!\in\!\left[ \Omega\!\left(\!\log\!\Lambda\!\right)\!,\!O(\!\polylog\!\Lambda\!)\!\right]\\
\Theta\!\left(\!\frac{K\!p\!(\!1\!-\!\gamma\!)}{1+t}\frac{\log \Lambda}{I\!p\log\! \frac{\log\!\Lambda}{Ip}}\!\right)\!, \text{ if } I\!p\!\in\!\left[ \Omega\left(\!\frac{1}{\polylog\!\Lambda}\!\right),  o(\!\log\!\Lambda\!)\right]\!,
\end{cases}
\end{align}
which can be further simplified as
\begin{align}
 \overline{T}\!(\!\mathbf{G}\!)\!=\!
 \begin{cases} 
\Theta\left( \frac{Kp(1-\gamma)}{1+t}\right)\!, \text{ if }  Ip  = \Omega\left(\log\Lambda \right)\\
\Theta\!\left(\!\frac{K\!p\!(\!1\!-\!\gamma\!)}{1+t}\frac{\log \Lambda}{I\!p\log\! \frac{\log\!\Lambda}{Ip}}\!\right)\!, \text{ if } I\!p\!\in\!\left[ \Omega\left(\!\frac{1}{\polylog\!\Lambda}\!\right),  o(\!\log\!\Lambda\!)\right]\!.
\end{cases}
\end{align}
This concludes the proof of Theorem \ref{th:lawtavg_iden}.

\subsection{Proof of Lemma \ref{le:tavg_op_lb}} \label{AP:tavg_op_lb}
From~\eqref{eq:tavg_data}, we know that $l_{s,\lambda}$ and ${\Lambda-\lambda \choose t}$ are non-increasing with $\lambda$, which implies that for each time slot $s\in [S]$, the profile vector $\mathbf{L}_s$, which minimizes the delay has components of the form
\begin{align}\label{eq:Ltbc}
l_{s,\lambda} =
\begin{cases} 
 \big\lfloor \frac{d_s}{\Lambda}\big\rfloor+1\!\!\!  \hspace{0.2cm} \text{ for }  \lambda \in \left[1, 2,  \dots, A_s\right]   \\
 \big\lfloor \frac{d_s}{\Lambda}\big\rfloor  \hspace{0.6cm} \text{ for }  \lambda \in \left[A_s\!+1, A_s\!+2, \dots, \Lambda\right],
\end{cases}
\end{align} 
where $d_s=\sum_{k\in [K]}d_{s,k}$, and $A_s \triangleq d_s- \Lambda \big\lfloor \frac{d_s}{\Lambda}\big\rfloor$. Consequently, when $A_s \geq \Lambda-t$, the corresponding best-case delay $T_{s}$ for time slot $s \in [S]$ is given by
\begin{align}\label{eq:tminc2}
T_{s} =  \sum_{\lambda=1}^{\Lambda-t} \left(\left\lfloor\frac{d_s}{\Lambda}\right\rfloor+1\right) \frac{ {\Lambda-\lambda \choose t} }{{\Lambda \choose t}} =   \left( \left\lfloor\frac{d_s}{\Lambda}\right\rfloor +1 \right)\frac{\Lambda-t}{1+t},
\end{align}
while when $A_s < \Lambda-t$, this is given as 
\begin{align}\label{eq:tminc1}
&T_{s} =    \sum_{\lambda=1}^{A_s} \left(\left\lfloor\frac{d_s}{\Lambda}\right\rfloor+1\right) \frac{ {\Lambda-\lambda \choose t} }{{\Lambda \choose t}}+\sum_{\lambda=A_s+1}^{\Lambda-t} \left\lfloor\frac{d_s}{\Lambda}\right\rfloor  \frac{ {\Lambda-\lambda \choose t} }{{\Lambda \choose t}}\nonumber\\
&=   \left\lfloor\frac{d_s}{\Lambda}\right\rfloor \sum_{\lambda=1}^{\Lambda-t}  \frac{ {\Lambda-\lambda \choose t} }{{\Lambda \choose t}} + \sum_{\lambda=1}^{A_s}  \frac{ {\Lambda-\lambda \choose t} }{{\Lambda \choose t}} \nonumber\\
& =    \left( \left\lfloor\frac{d_s}{\Lambda}\right\rfloor +1 \right) \sum_{\lambda=1}^{\Lambda-t} \frac{ {\Lambda-\lambda \choose t} }{{\Lambda \choose t}}  -\sum_{\lambda=A+1}^{ \Lambda-t }  \frac{ {\Lambda-\lambda \choose t} }{{\Lambda \choose t}} \nonumber\\
&=   \left( \left\lfloor\frac{d_s}{\Lambda}\right\rfloor +1 \right) \frac{  {\Lambda \choose t+1}}{{\Lambda \choose t}}  - \frac{ {\Lambda-A_s \choose t+1}}{{\Lambda \choose t}} \nonumber\\
& =\left( \left\lfloor\frac{d_s}{\Lambda}\right\rfloor +1 \right)\frac{\Lambda-t}{1+t}- \frac{ {\Lambda-A_s \choose t+1}}{{\Lambda \choose t}}.
\end{align}
We denote $\mathbf{S}_2 \subseteq [S]$ to be the set of time slots for which $A_s < \Lambda-t$. Then, the average delay corresponding to the optimal user-to-cache state association $\hat{\mathbf{G}}$ is lower bounded by 
\begin{align} \label{eq:opt_lb}
  \overline{T}^* &\geq \frac{1}{S} \sum_{s\in [S]/\mathbf{S}_2} \left( \left\lfloor\frac{d_s}{\Lambda}\right\rfloor +1 \right)\frac{\Lambda-t}{1+t}\nonumber\\ &+  \frac{1}{S}  \sum_{s\in [\mathbf{S}_2]} \left(\left( \left\lfloor\frac{d_s}{\Lambda}\right\rfloor +1 \right)\frac{\Lambda-t}{1+t}- \frac{ {\Lambda-A_s \choose t+1}}{{\Lambda \choose t}}\right)    \nonumber\\
 & =\frac{1}{S} \sum_{s\in [S]} \left( \left\lfloor\frac{d_s}{\Lambda}\right\rfloor +1 \right)\frac{\Lambda-t}{1+t}  -\frac{1}{S}   \sum_{s\in [\mathbf{S}_2]}  \frac{ {\Lambda-A_s \choose t+1}}{{\Lambda \choose t}}.   
\end{align}
 This concludes the proof of Lemma \ref{le:tavg_op_lb}.

\subsection{Proof of Theorem \ref{th:tavg_alg1_ub}} \label{AP:tavg_alg1_ub}

We denote $v^{\mathcal{G}_1}_{s,\lambda}$, $v^{\mathcal{G}_2}_{s,\lambda}$, and $v^{\mathcal{G}}_{s,\lambda}$ to be the scaled loads calculated using transformed user demand matrix $\mathbf{D}$ (Step 00 of \textbf{Algorithm \ref{alg:algo_1}}) of each cache state $\lambda\in[\Lambda]$ at time slot $s\in[S]$ following the user-to-cache state association given by $\mathcal{G}_1$, $\mathcal{G}_2$, and $\mathcal{G}$ respectively. It is straightforward to see from step 03 of \textbf{Algorithm \ref{alg:algo_1}} that $v^{\mathcal{G}_1}_{s,\lambda} =O \left(\frac{\log S}{\log \log S} \right)$ $\forall s\in [S],  \ \lambda \in [\Lambda]$. By combining Lemma 15 and Lemma 18 of \cite{im_focs15}, we have $v^{\mathcal{G}_2}_{s,\lambda}=O \left(1\right)$ $\forall s\in [S],  \ \lambda \in [\Lambda]$. Thus, the combined scaled load of each cache $\lambda \in [\Lambda]$ at each time slot $s \in[S]$ is given by $v^{\mathcal{G}}_{s,\lambda}=O \left(\frac{\log S}{\log \log S} \right).$ 

To complete the proof, we now proceed to convert the scaled load of each cache $\lambda\in[\Lambda]$ at time slot $s\in[S]$ to the actual load. Based on the assumption that for each time slot $s\in[S]$,  $\sum_{k\in[K]} d_{s,k} \geq \Lambda$, we have  $\bar{d}_{s,k}  = \min\left(\frac{\Lambda \ d_{s,k}}{\sum_{i\in[K]} d_{s,i}},1\right) = \frac{\Lambda \ d_{s,k}}{\sum_{i\in[K]} d_{s,i}}$ $\forall \ s\in [S], k\in[K]$. Then, the actual load corresponding to user-to-cache state association $\mathcal{G}$ is given by
 \begin{align}
 v_{s,\lambda}  = \frac{\sum_{k\in[K]} d_{s,k}}{\Lambda}v^{\mathcal{G}}_{s,\lambda}=  O \left(\frac{\sum_{k\in[K]} d_{s,k}}{\Lambda} \frac{\log S}{\log \log S} \right) \nonumber
 \end{align}
$\forall\ s\in [S],\   \lambda \in [\Lambda]$. Consequently, from \eqref{eq:tavg_data}, we have
\begin{align}
  \overline{T}(\mathbf{G}) &= O\left(\frac{1}{S}\sum_{s=1}^{S}\sum_{\lambda=1}^{\Lambda-t}   \frac{\sum_{k\in[K]} d_{s,k}}{\Lambda} \frac{\log S}{\log \log S}   \frac{{\Lambda-\lambda \choose t}}{{\Lambda \choose t}}\right)\\
  &= O\left(\frac{\log S}{\log \log S}\frac{1}{S}\sum_{s=1}^{S}  \frac{\sum_{k\in[K]} d_{s,k}}{\Lambda}  \frac{\Lambda-t}{1+t}\right).
\end{align}
 We also have from \eqref{eq:opt_lb} that
 \begin{align}
  \overline{T}^* =\Omega\left(\frac{1}{S} \sum_{s\in [S]} \frac{\sum_{k\in[K]} d_{s,k}}{\Lambda} \frac{\Lambda-t}{1+t}\right).   
\end{align}
This concludes the proof of Lemma \ref{le:tavg_op_lb}.


\bibliographystyle{IEEEtran}
\bibliography{IEEEabrv,main}

\end{document}